\pdfoutput=1
\documentclass[final,1p,times,authoryear,nopreprintline]{elsarticle}

\def\titl{Verification Protocols with Sub-Linear Communication for Polynomial Matrix Operations}

\PassOptionsToPackage{ruled,boxed,linesnumbered,vlined,noend}{algorithm2e}
\PassOptionsToPackage{nameinlink,capitalize,noabbrev}{cleveref}
\usepackage{%
  algorithm2e,
  amsmath,amsthm,bm,%
  multirow,newfloat,tabularx,%
  tikz,booktabs,%
  rotating,%
  tcolorbox,caption,%
  hyperref,cleveref
}

\hypersetup{
  pdftitle={\titl},
  colorlinks,
  linkcolor=DarkBlue,
  citecolor=DarkGreen,
  urlcolor=DarkBlue,
  bookmarksnumbered,
  pdfauthor=author,
}

\DontPrintSemicolon
\SetKw{break}{break}

\defcitealias{GuSaStVa12}{\emph{ibid.}}

\newtheorem{theorem}{Theorem}[section]
\newtheorem{fact}[theorem]{Fact}
\newtheorem{lemma}[theorem]{Lemma}
\newtheorem{corollary}[theorem]{Corollary}
\newtheorem{definition}[theorem]{Definition}
\numberwithin{equation}{section}

\DeclareMathOperator{\lcm}{lcm}

\newcolumntype{Y}{>{\centering\arraybackslash}X}

\newcommand{\ZZ}{\ensuremath{\mathbb{Z}}}
\newcommand{\NN}{\ensuremath{\mathbb{N}}}

\newcommand{\defeq}{=} 

\newcommand{\F}{\ensuremath{\mathsf{F}}}

\newcommand{\matr}[1]{\ensuremath{\bm{#1}}}
\newcommand{\vect}[1]{\matr{#1}}
\newcommand{\transpose}[1]{\ensuremath{#1^\mathsf{T}}}
\newcommand{\zeromat}{\mathbf{0}}
\newcommand{\zerovect}{\mathbf{0}}
\DeclareMathOperator{\rank}{rank} 
\DeclareMathOperator{\rowsp}{RowSp}
\DeclareMathOperator{\colsp}{ColSp}
\DeclareMathOperator{\satur}{Saturation}
\DeclareMathOperator{\denom}{denom}
\DeclareMathOperator{\numer}{numer}

\newcommand{\fieldsubset}{\textsf{S}\xspace}
\newcommand{\bnd}[2]{\ensuremath{#1\mathopen{}\left(#2\right)\mathclose{}}}
\newcommand{\oh}[1]{\bnd{O}{#1}}
\newcommand{\softoh}[1]{\bnd{\widetilde{O}}{#1}}
\newcommand{\tsbnd}[2]{\ensuremath{#1\mathopen{}(#2)\mathclose{}}} 
\newcommand{\tsoh}[1]{\tsbnd{O}{#1}} 
\newcommand{\tssoftoh}[1]{\tsbnd{\widetilde{O}}{#1}} 
\newcommand{\ceil}[1]{\ensuremath{\left\lceil#1\right\rceil}}

\DeclareFloatingEnvironment[placement={htbp}]{protocolfloat}

\newcounter{cprotocolsteps}
\newenvironment{protocol}{%
  \setcounter{cprotocolsteps}{0}
  \begin{protocolfloat}%
    \begin{tcolorbox}%
}{%
    \end{tcolorbox}%
  \end{protocolfloat}%
}
\crefname{protocolfloat}{Protocol}{Protocols}
\crefname{cprotocolsteps}{Step}{Steps}
\newlength{\protowid}
\newcommand{\Public}[1]{\textbf{Public}: #1 \par}
\newcommand{\Certifies}[1]{\textbf{Certifies}: #1\par}
\newcommand\step{\\ \refstepcounter{cprotocolsteps} \thecprotocolsteps. }
\newenvironment{protocolsteps}{%
  \vspace{0.2cm}
  \tabularx{\textwidth}{@{}lX@{}Y@{}X@{}}
  & \multicolumn{1}{c}{\textbf{Prover}} &&
    \multicolumn{1}{c}{\textbf{Verifier}} \\
  \cline{2-2} \cline{4-4}
}{%
  \endtabularx%
}

\newcommand{\prover}[1]{& #1}
\newcommand{\verifier}[1]{&&& #1}
\newcommand{\toverifier}[1]{\\ && $\xrightarrow[\rule{\protowid}{0ex}]{\makebox{#1}}$}
\newcommand{\toprover}[1]  {\\ && $\xleftarrow[\rule{\protowid}{0ex}]{\makebox{#1}}$}
\newcommand{\subprotocol}[1]{%
  & \multicolumn{3}{c}{\tikz[baseline=-5pt]{\node[draw,dashed] {#1};} }}
\newcommand{\checkeq}{\stackrel{?}{=}}
\newcommand{\checkneq}{\stackrel{?}{\neq}}
\newcommand{\checklt}{\stackrel{?}{<}}
\newcommand{\checkle}{\stackrel{?}{\le}}

\newcommand{\checkge}{\stackrel{?}{\ge}}
\newcommand{\checkin}{\stackrel{?}{\in}}
\newcommand{\checksub}{\stackrel{?}{\subseteq}}
\newcommand{\randfrom}{\xleftarrow{\$}}
\newcommand{\lowrank}{\texttt{\textup{LOW\_RANK}}}
\newcommand{\nosol}{\texttt{\textup{NO\_SOLUTION}}}

\newcommand{\singular}{\hyperref[pro:singular]{\textsf{Singularity}}}
\newcommand{\nonsingular}{\hyperref[pro:nonsingular]{\textsf{NonSingularity}}}
\newcommand{\ranklb}{\hyperref[pro:ranklb]{\textsf{RankLowerBound}}}
\newcommand{\rankub}{\hyperref[pro:rankub]{\textsf{RankUpperBound}}}
\newcommand{\prorank}{\hyperref[pro:rank]{\textsf{Rank}}}
\newcommand{\prodet}{\hyperref[pro:det]{\textsf{Determinant}}}
\newcommand{\systemsolve}{\hyperref[pro:systemsolve]{\textsf{SystemSolve}}}
\newcommand{\matmul}{\hyperref[pro:matmul]{\textsf{MatMul}}}
\newcommand{\rowmemf}{\hyperref[pro:rowmemf]{\textsf{FullRankRowSpaceMembership}}}
\newcommand{\progcd}{\hyperref[pro:gcd]{\textsf{CoPrime}}}
\newcommand{\rowmem}{\hyperref[pro:rowmem]{\textsf{RowSpaceMembership}}}
\newcommand{\rowsubset}{\hyperref[pro:rowsubset]{\textsf{RowSpaceSubset}}}
\newcommand{\rowspaceeq}{\hyperref[pro:rowspaceeq]{\textsf{RowSpaceEquality}}}
\newcommand{\rowbasis}{\hyperref[pro:rowbasis]{\textsf{RowBasis}}}
\newcommand{\saturated}{\hyperref[pro:saturated]{\textsf{Saturated}}}
\newcommand{\satbasis}{\hyperref[pro:satbasis]{\textsf{SaturationBasis}}}
\newcommand{\ishermite}{\hyperref[pro:hermite]{\textsf{HermiteForm}}}
\newcommand{\ispopov}{\hyperref[pro:spopov]{\textsf{ShiftedPopovForm}}}
\newcommand{\completable}{\hyperref[pro:completable]{\textsf{UnimodularCompletable}}}
\newcommand{\kernelbasis}{\hyperref[pro:kernel]{\textsf{KernelBasis}}}

\newcommand{\field}{\F} 
\newcommand{\var}{x} 
\newcommand{\polRing}{\field[x]} 
\newcommand{\fracRing}{\field(x)} 
\newcommand{\matRing}[2]{\field^{#1 \times #2}} 
\newcommand{\pmatRing}[2]{\polRing^{#1 \times #2}} 
\newcommand{\idMat}[1]{\matr{I}_{#1}} 
\newcommand{\shift}{\vect{s}} 
\newcommand{\hermite}{\matr{B}} 
\newcommand{\popov}{\matr{B}} 
\newcommand{\vecspace}{\mathsf{V}} 
\newcommand{\dd}{f} 

\begin{document}

\begin{frontmatter}

  \title{\titl}

\author[1]{David Lucas}
\ead{David.Lucas@univ-grenoble-alpes.fr}

\author[2]{Vincent Neiger}
\ead{vincent.neiger@unilim.fr}

\author[1]{Cl\'ement Pernet}
\ead{Clement.Pernet@univ-grenoble-alpes.fr}

\author[3]{Daniel S. Roche\corref{cor1}}
\ead{roche@usna.edu}

\author[4]{Johan~Rosenkilde}
\ead{jsrn@dtu.dk}

\address[1]{Univ. Grenoble Alpes, CNRS, Grenoble INP, LJK, 38000 Grenoble, France}
\address[2]{Univ. Limoges, CNRS, XLIM, UMR 7252, F-87000 Limoges, France}
\address[3]{United States Naval Academy, Annapolis, Maryland, U.S.A.}
\address[4]{Technical University of Denmark, Kgs.\ Lyngby, Denmark}

\cortext[cor1]{Corresponding author}

\begin{abstract}
  We design and analyze new protocols to verify the correctness of various
  computations on matrices over the ring $\F[x]$ of univariate polynomials over
  a field $\F$. For the sake of efficiency, and because many of the properties
  we verify are specific to matrices over a principal ideal domain, we cannot
  simply rely on previously-developed linear algebra protocols for matrices
  over a field.  Our protocols are \emph{interactive}, often randomized, and
  feature a constant number of rounds of communication between the Prover and
  Verifier. We seek to minimize the communication cost so that the amount of
  data sent during the protocol is significantly smaller than the size of the
  result being verified, which can be useful when combining protocols or in
  some multi-party settings.  The main tools we use are reductions to existing
  linear algebra verification protocols and a new protocol to verify that a
  given vector is in the $\F[x]$-row space of a given matrix.
\end{abstract}
\end{frontmatter}

\section{Introduction}

Increasingly, users or institutions with large computational needs are
relying on untrusted sources of computational results, which could be remote
(``cloud'') servers, unreliable hardware, or even just Monte Carlo randomized
algorithms. The rising area of \emph{verifiable computing} seeks to
maintain the benefits in cost or speed of using such untrusted sources,
without sacrificing accuracy. Generally speaking, the goal is to
develop \emph{protocols} certifying the correctness of some result, which
can be verified much more efficiently than re-computing the result
itself.

\subsection{Verification protocols}
In this paper, we
propose new \emph{verification protocols} for computations performed on univariate
polynomial matrices; we refer to \citep{Dum18, dk14, kns11, klyz11} for definitions
related to such protocols.
Generically, we consider protocols where a
Prover performs computations and provides additional data structures to or
exchanges with a Verifier, who will use these to check the validity of a result,
at a lower cost than by recomputing it.
From the viewpoint of the theoretical computer science community, this corresponds to instances of
interactive proof protocols \citep{GMR89,gkr08}. In the context of certified compter algebra, the
term interactive proof of work certificate was introduced by \citet{Dum18}. In this paper, we will
simply refer to these protocols as verification protocols, as they do not 
necessarily verify the work of the  computation  but only the result (which can be obtained by any means).

The general flow of a verification protocol is as follows.
\begin{enumerate}
\item At the beginning, the Prover and the Verifier share the knowledge of the result of a
  computation, which they have to verify.
\item The Verifier sends  a Challenge to the Prover, usually consisting of
some uniformly sampled random values.
\item The Prover replies with a Response, used
by the Verifier to ensure the validity of the commitment.
\item In some cases,
several additional rounds of Challenge/Response might be necessary for the Verifier to
accept an answer.
\end{enumerate}

These protocols can be simulated non-interactively in a single round
following the heuristic derandomization of \citet{fiatshamir86}: random values
produced by the Verifier are replaced by cryptographic hashes of the input and
previous messages, and the Prover publishes once both the Commitment and
Response to the derandomized Challenge.

There are several metrics to assess the efficiency of a verification protocol, namely
\begin{description}
  \item [Communication cost:] the volume of data exchanged throughout the
    protocol;
  \item [Verifier cost:] the worst-case number of arithmetic operations
    performed by the Verifier in the protocol, no matter what data was
    sent by the Prover;
  \item [Prover cost:] the number of arithmetic operations performed by
    an \emph{honest Prover} that is trying to prove a statement which is
    actually true without fooling the Verifier.
\end{description}
Note that some data, namely the input and output to the original
problem, are considered as \emph{public data} and do not count towards
the communication cost. This is to remove those parts which are somehow
inherent in the problem itself, as well as to separate the functions of
computing and verifying a result, which can be quite useful when
verification protocols are combined, as we will see.

Such protocols are said complete if the probability that a true statement is
rejected by a Verifier can be made arbitrarily small; they are said perfectly
complete if true statements are never rejected. For simplicity's sake, as all
the protocols in this paper are perfectly complete, we will sometimes just
describe them as complete. Similarly, a protocol is sound if the probability
that a false statement is accepted by the Verifier can be made arbitrarily
small. Note that all our protocols are probabilistically sound, which means
there is a small probability the
Verifier may be tricked into accepting a wrong answer. This is not an issue, as in
practice this probability can be reduced by simply repeating the protocol with
new randomness, or by computing over a larger field.
As our protocols are perfectly complete, any single failure
means that the statement is false and/or the Prover did something wrong;
the Verifier or unlucky randomness are never to blame.

Several approaches to verified computation exist: generic approaches based on
protocol check circuits \citep{gkr08} or on homomorphic
encryption \citep{cfhkknpz15}; and approaches working for any algorithm where the
Prover uses specific operations, such as that of \citet[Section 5]{kns11} which certifies
any protocol where matrix multiplications are
performed. Another approach consists in designing problem-specific
verification protocols, as was done for instance in \citep{Fre79, kns11, dlp17}
on dense linear algebra and
\citep{dk14, dktv16} on sparse linear algebra.

\subsection{Polynomial matrices}

This paper deals with computations on matrices whose entries are univariate polynomials.
While certification for matrices over fields and over integer rings have been
studied over the past twenty years, there are only few results on polynomial
matrices \citep{Dum18,gn18}.

A \emph{polynomial matrix} is a matrix
\(\matr{M}\in\pmatRing{m}{n}\) whose entries are univariate
polynomials over a field \(\field\). There is an isomorphism with \emph{matrix
polynomials} (univariate polynomials with matrices as
coefficients) which we will sometimes use implicitly, such as when
considering the evaluation \(\matr{M}(\alpha) \in \matRing{m}{n}\) of \(\matr{M}\)
at a point \(\alpha\in\field\).

Computations with polynomial matrices are of central importance in
computer algebra and symbolic computation, and many efficient algorithms
for polynomial matrix computations have been developed.

One general approach for computing with polynomial matrices is based on
evaluation and interpolation. The basic idea is to first evaluate the
polynomial matrix, say \(\matr{M}\in\pmatRing{n}{n}\) at a set of points
\(\alpha_1,\alpha_2,\ldots\in\field\) in the ground field, then to separately
perform the desired computation on each \(\matr{M}(\alpha_i)\) over
\(\matRing{n}{n}\), and finally reconstruct the entries of the result using
fast polynomial interpolation. This kind of approach works well for
operations such as matrix multiplication \citep[Section 5.4]{BosSch05} or
determinant computation. These computations essentially concern the
\emph{vector space} in the sense that \(\matr{M}\) may as well be seen as a
matrix over the fractions \(\fracRing\) without impact on the results of the
computations.

Other computational problems with polynomial matrices intrinsically concern
\emph{\(\polRing\)-modules} and thus cannot merely rely on evaluation and
interpolation. Classic and important such examples are that of computing normal
forms such as the Popov form and the Hermite form
\citep{Popov72,Villard96,NRS18} and that of computing modules of relations such
as approximant bases \citep{BecLab94,GiJeVi03,NeiVu17}. The algorithms in this
case must preserve the module structure attached to the matrix and thus deal
with the actual polynomials in some way; in particular, an algorithm which
works only with evaluations of the matrix at points \(\alpha\in\field\) is oblivious
to this module structure.

\subsection{Our contributions}

In this paper, after giving some preliminary material in~\cref{sec:prelim}, we propose verification
protocols
for classical properties of polynomial matrices --- singularity, rank, determinant and matrix
product --- with sub-linear communication cost with respect to the input size (\cref{sec:vspace}).
Those protocols are based on evaluating considered matrices at random points, which allows us
to reduce the communication space and to use existing verification protocols for matrices over fields.
Then, in~\cref{sec:rowmem} we give the main result of this paper, which is 
certifying that a given polynomial row vector is in the row space of a given
polynomial matrix, which can either have full rank or be rank-deficient.
\cref{sec:rowspace_normalforms} shows how to use this result to certify that for two given
polynomial matrices $\matr{A}$ and $\matr{B}$, the row space of $\matr{A}$ is contained
in the row space of $\matr{B}$, and then gives verification protocols for some classical
normal forms of polynomial matrices. In \cref{sec:saturation}, we present verification protocols related to 
saturations and kernels of polynomial matrices. Finally, \cref{sec:perspectives}
gives a conclusion and comments on a few perspectives.

A summary of our contributions is given in~\cref{tab:contributions}, based on
the following notations: the input matrix has rank $r$ and size $n \times n$ if
it is square or $m \times n$ if it can be rectangular; if there are several input
matrices, then \(r\) stands for the maximum of their ranks, \(m\) for the
maximum of their row dimensions, and \(n\) for the maximum of their column
dimensions. Where appropriate, \(r\) is the maximum of the actual ranks of the matrices
and the claimed rank by the Prover. We write $d$ for the maximum degree of
any input matrix or vector.

The Prover and Verifier costs are in arithmetic operations over the base
field \(\field\).
We use \(\softoh{\cdot}\) for asymptotic cost
bounds with hidden logarithmic factors, and \(\omega \le 3\) is the exponent
of matrix multiplication, so that the multiplication of two \(n\times n\) matrices over
\(\field\) uses \(\tsoh{n^\omega}\) operations in \(\field\); see
\Cref{sec:prelim} for more details and references.

The last column indicates the smallest
size of the ground field \(\field\)
needed to ensure both perfect completeness of the protocol and soundness with
probability at least \(\tfrac{1}{2}\).
If this lower bound is not met, an extension
field may be agreed on in advance by the Prover and Verifier, for a
(logarithmic) increase in arithmetic and communication costs.
For all protocols, an arbitrary low probability $p$ of failure can be achieved
by simply iterating the protocol at most \(\ceil{\log_2(1/p)}\) times.

\begin{table}[htbp]\small
\begin{center}
  \begin{sideways}
    \begin{tabular}{llllll}
        \toprule
        & \multicolumn{2}{l}{Prover}
        & \multirow{2}{*}{Comm.} & Verifier & \multirow{2}{*}{Minimum \(\#\field\)} \\
        \cmidrule{2-3}
        &  Deter. & Cost & & Cost & \\
        \midrule
      \singular{} &  Yes & $\tsoh{nr^{\omega-1} + n^2d}$ & $\tsoh{n}$ & $\tsoh{n^2 d}$ & $2nd$ \\[0.05cm]
      \nonsingular{}  & Yes  & $\tssoftoh{n^\omega d}$ & $\tsoh{n}$ & $\tsoh{n^2 d}$ & $nd+1$ \\[0.05cm]
      \ranklb{}  & No & $\tsoh{mnr^{\omega-2} + mnd}$ & $\tsoh{r}$ & $\tsoh{r^2 d}$ & $rd+1$ \\[0.05cm]
      \rankub{}  & Yes & $\tsoh{mnr^{\omega-2}+ mnd}$ & $\tsoh{n}$ & $\tsoh{mnd}$ & $2rd+2$ \\[0.05cm]
      \prorank{} & No & $\tsoh{mnr^{\omega-2} + mnd}$ & $\tsoh{n}$ & $\tsoh{mnd}$ & $2rd+2$ \\[0.05cm]
      \prodet{}  & Yes & $\tsoh{n^{2}d + n^\omega}$ & $\tsoh{n}$ & $\tsoh{n^2 d}$ & $2nd+2$ \\[0.05cm]
      \systemsolve & N/A & N/A & 0 & $\tsoh{n^2 d}$ & $4d$ \\[0.05cm]
      \matmul{} & N/A & N/A & 0 & $\tsoh{n^2d}$ & $4d+2$ \\[0.05cm]
      \rowmemf{}  & Yes & $\tssoftoh{nm^{\omega-1}d}$ & $\tsoh{md}$ & $\tsoh{mnd}$ & $6md+2d+2$\\[0.05cm]
      \rowmem{}  & No & $\tssoftoh{mnr^{\omega -2}d}$ & $\tssoftoh{md+n}$ & $\tssoftoh{mnd}$ & $6md+2d+2$ \\[0.05cm]
      \rowsubset{}  & No &$\tssoftoh{mnr^{\omega -2}d}$ &$\tssoftoh{md+n}$ & $\tssoftoh{mnd}$&$8rd+2d+4$\\[0.05cm]
      \rowspaceeq{}  & No &$\tssoftoh{mnr^{\omega -2}d}$ &$\tssoftoh{md+n}$ &$\tssoftoh{mnd}$&$8rd+2d+4$ \\[0.05cm]
      \rowbasis{}  & No & $\tssoftoh{mnr^{\omega -2}d}$&$\tssoftoh{md+n}$ & $\tssoftoh{mnd}$&$8rd+2d+6$\\[0.05cm]
      \ishermite   & No & $\tssoftoh{mnr^{\omega-2}d}$ & $\tssoftoh{md+n}$ & $\tssoftoh{mnd}$ & $8rd+2d+4$ \\[0.05cm]
      \ispopov     & No & $\tssoftoh{mnr^{\omega-2}d}$ & $\tssoftoh{md+n}$ & $\tssoftoh{mnd}$ & $8rd+2d+4$ \\[0.05cm]
      \saturated{} (\(m\le n\))  & No & $\tssoftoh{n m^{\omega-1} d}$ & $\tssoftoh{n d}$ & $\tssoftoh{m n d}$ & $8md+4$ \\[0.05cm]
      \saturated{} (\(m > n\)) & No & $\tssoftoh{m n^{\omega-1} d}$ & $\tssoftoh{m d}$ & $\tssoftoh{m n d}$ & $8n d+4$ \\[0.05cm]
      \satbasis  & No & $ \tssoftoh{m n r^{\omega-2} + m n d + n r^{\omega-1} d}$ & $\tssoftoh{n d}$ & $\tssoftoh{m n d}$ & $8nd + 2d +4$ \\[0.05cm]
      \completable  & No & $ \tssoftoh{n m^{\omega-1} d}$ & $\tssoftoh{n d}$ & $\tssoftoh{m n d}$ & $8md+4$ \\[0.05cm]
        \kernelbasis  & No & $ \tssoftoh{(m+n)m^{\omega-1}d}$ & $\tssoftoh{md}$ & $\tssoftoh{m(m+n)d}$ & $8md+4$ \\
        \bottomrule
    \end{tabular}
\end{sideways}
\end{center}
    \caption{Summary of the contributions. The first column states whether the
      Prover's algorithm is deterministic or not. The costs are given in number
      of arithmetic operations over the base field and the communication is in
      number of elements in the base field \(\field\). The last column reports
    the minimum size size of \(\field\) needed to ensure perfect completeness
  and soundness with probability at least \(\tfrac{1}{2}\). }
    \label{tab:contributions}
\end{table}

\section{Preliminaries}\label{sec:prelim}

\subsection{Notation and assumptions}

\paragraph{Fields and rings}
We use $\F$ to indicate an arbitrary field, $\F[x]$ for the ring of polynomials
in one variable $x$ with coefficients in $\F$, and $\F(x)$ for the field of
rational fractions, i.e., the fraction field of $\F[x]$. The ring of \(m\times
n\) matrices, for example over \(\polRing\), is denoted by \(\pmatRing{m}{n}\).

\paragraph{Asymptotic complexity bounds}
Throughout the paper, the cost bounds are worst-case deterministic unless otherwise indicated.
We use the ``soft-oh'' notation \(\softoh{\cdot}\) to give asymptotic bounds
hiding logarithmic factors. Precisely, for two cost functions
\(f,g\), having \(f\in\softoh{g}\) means that \(f\in\oh{g\log(g)^c}\)
for some constant \(c > 0\).

We write \(\omega\) for the exponent of matrix multiplication over
\field, so that any two matrices \(\matr{A},\matr{B}\in\matRing{n}{n}\)
can be multiplied using \(\oh{n^\omega}\) field operations;
we have \(2\le \omega\le
3\) and one may take \(\omega < 2.373\) \citep{CopWin90,LeGall14}.

\citet{CanKal91} have shown that multiplying two univariate polynomials of
degree \(\le d\) over any algebra uses \(\softoh{d}\) additions,
multiplications, and divisions in that algebra. In particular, multiplying two
matrices in \(\pmatRing{n}{n}\) of degree at most \(d\) uses \(\softoh{n^\omega
d}\) operations in \(\field\).

\paragraph{Sampling set}

In our protocols, \(\fieldsubset\) is always a finite subset of the base field
\(\field\) which the Verifier uses to sample field elements uniformly and independently
at random.
We denote by
\[
  \alpha \randfrom \fieldsubset
  \quad\text{and}\quad
  \vect{v} \randfrom \fieldsubset^{n \times 1} 
\]
respectively the actions of drawing a field element uniformly at random from
\(\fieldsubset\) and of drawing a vector of \(n\) field elements uniformly
and independently at random from \(\fieldsubset\).

To ensure that they are perfectly complete and probabilistically sound,
our protocols require lower bounds
on the cardinality \(\#\fieldsubset\) of this subset, and therefore of
the field \(\field\) itself. Generally speaking, choosing
\(\fieldsubset\) larger will increase the soundness probability, at the
cost of higher randomness complexity. In particular, one may use
\(\fieldsubset=\field\) if the field \(\field\) is finite and
sufficiently large. If \(\#\field\) is too small,
then one may use a field extension, causing up to a logarithmic factor
increase in the
Prover/Verifier/communication costs.

\paragraph{Protocols}

Many of our analyses and protocols use the notation
\[
  d_{\matr{A}} = \max(1,\deg(\matr{A}))
  \quad\text{and}\quad
  r_{\matr{A}} = \rank(\matr{A})
\]
for any polynomial matrix \(\matr{A}\) that appears in a given protocol.

Following \citep{dk14,dktv16,dlp17} we use the notation
$x\checkeq y$ as a placeholder for
\begin{center}
  \textbf{If} \(x \neq y\) \textbf{then} abort and report failure
\end{center}
to improve the brevity and readability of the protocols. Similarly, we use $x
\checklt y, x\checkge y$, $U\checksub V$, etc.~to check inequalities and set
inclusion.

In all our protocols, we implicitly assume that the Verifier always checks whether the data
received has the correct type and size, and aborts immediately if such
checks fail. For instance, in \Cref{pro:ranklb}, the Verifier implicitly
checks that the
subsets $I$ and $J$ are actually subsets of \(\rho\) distinct integers
each, as specified by the protocol, but explicitly checks they are subsets of
\(\{1,\dots, m\}\) and \(\{1,\dots, n\}\) respectively.

\paragraph{Threat model}

As stated before, our protocols are perfectly complete and
probabilistically sound, meaning that (1) if the statement is true and
the Prover is honest, then an honest Verifier will always accept the
statement; and (2) if the statement is not true, then an honest Verifier
will reject the statement with probability at least \(\frac{1}{2}\).

In practice, this chance of incorrectly accepting a false statement can
be made arbitrarily low by iterating the protocol; namely, the
probability of wrongly accepting a false statement is less than
\(2^{-\lambda}\) if the protocol succeeds in every one of \(\lambda\)
iterations. Note that these iterations can always be performed
in parallel, so that there is a linear scaling in the communication,
Verifier, and Prover costs, but not in the number of rounds in the
protocols.

The probabilistic soundness holds whenever the Verifier is
honest; the Prover may be malicious and deviate from the stated
protocol. Our security depends only on the random challenges chosen by
the Verifier and various algebraic properties; we do not make any
cryptographic assumptions or impose
limits on the computational power of the Prover.

The cost of an honest Verifier depends only on the public information
--- that is, the parameters of the statement being verified --- and even a
malicious Prover cannot cause the Verifier to do more work than this.
However, there is an issue here in the case of very large (or
even infinite) fields \(\field\), where the Prover may essentially
execute a denial of service attack by sending arbitrary large field
elements in the protocol. Because our analysis is generic in terms of field
operations, it is not able to capture this weakness, and we do not
attempt to address it.

\paragraph{Rational fractions}
For a rational fraction $f\in\F(x)$, define its \emph{denominator} $\denom(f)$ to be the unique
monic polynomial $g\in\F[x]$ of minimal degree such that $gf \in \F[x]$.
Correspondingly, define its \emph{numerator} $\numer(f) \defeq f\cdot\denom(f)$.
Note that $\denom(a)=1$ if and only if $a\in\polRing$.
More generally, for a matrix of rational fractions
$\matr{A}\in\F(x)^{m\times n}$, define $\denom(\matr{A})$ to be the unique monic
polynomial $g\in\F[x]$ of minimal degree such that
$g\matr{A} \in \F[x]^{m\times n}$, and again write this polynomial
matrix $g\matr{A}$ as
$\numer(\matr{A})$. Note that we have the identity
$\denom(\matr{A}) = \lcm_{i,j}(\denom(\matr{A}_{i,j}))$.

\paragraph{Row space, kernel, and row basis}
For a given matrix \(\matr{A} \in \pmatRing{m}{n}\), two basic
sets associated to it are its row space
\[
  \rowsp_{\polRing}(\matr{A}) \defeq \{\vect{p}\matr{A},\, \vect{p}\in\pmatRing{1}{m}\},
\]
and its left kernel
\[
  \{\vect{p} \in \pmatRing{1}{m} \mid \vect{p}\matr{A} = \zeromat\}.
\]
Accordingly, a \emph{row basis} of \(\matr{A}\) is a matrix in
\(\pmatRing{r}{n}\) whose rows form a basis of the former set, where \(r\) is
the rank of \(\matr{A}\), while a \emph{left kernel basis} of \(\matr{A}\) is a
matrix in \(\pmatRing{(m-r)}{n}\) whose rows form a basis of the latter set.
We use similar notions and notations for column spaces and column bases, and
for right kernels and right kernel bases. We will also often consider
the \emph{rational} row space or \(\fracRing\)-row space of \(\matr{A}\), denoted by
\(\rowsp_{\fracRing}(\matr{A})\), which is an \(\fracRing\)-vector space.

Matrices which preserve the row space under left-multiplication, that is,
\(\matr{U} \in \pmatRing{m}{m}\) such that the \(\polRing\)-row space of
\(\matr{U}\matr{A}\) is the same as that of \(\matr{A}\), are said to be
\emph{unimodular}. They are characterized by the fact that their determinant is
a nonzero constant; or equivalently that they have an inverse (with entries in
\(\polRing\)).

\subsection{Some probability bounds.}

Many of our protocols rely on the fact that when picking an element uniformly
at random from a sufficiently large finite subset of the field, this element is
unlikely to be a root of some given polynomial. This was stated formally in
\citep{Schwartz80,Zippel79,DeMilloLipton78} and is often referred to as the
\emph{DeMillo-Lipton-Schwartz-Zippel lemma}.

Specifically, it states that for any nonzero $k$-variate polynomial
\(f(x_1,\ldots,x_k)\) with coefficients in a field \(\field\), and any finite subset
\(\fieldsubset\subseteq\field\), if an evaluation point
\((\alpha_1,\ldots,\alpha_k)\in\field^k\) has entries chosen at
random uniformly and independently from \(\fieldsubset\), then the probability that
\(f(\alpha_1,\ldots,\alpha_k)=0\) is at most \(d/\#\fieldsubset\), where
\(d\) is the total degree of $f$.

The following consequence is a standard extension of the soundness proof of
Freivalds' algorithm \citeyearpar{Fre79}.
\begin{lemma}\label{lem:randinnerprod}
  Let $\matr{A}\in\F^{m\times n}$ be a matrix with at least
  one nonzero entry
  and let $\fieldsubset\subseteq\F$ be a finite subset.
  For a vector of scalars $\vect{w}\in\fieldsubset^{n\times 1}$ chosen
  uniformly at random, we have
  $\Pr[\matr{A}\vect{w}=\zerovect] \le 1/\#\fieldsubset$.
\end{lemma}
\begin{proof}
  Consider each of the $n$ entries of $\vect{w}$ as an indeterminate.
  Because $\matr{A}$ is not zero, $\matr{A}\vect{w}$ has at least one
  nonzero entry, which is a nonzero polynomial in $n$
  variables with total degree 1. Then a direct application of the
  DeMillo-Lipton-Schwartz-Zippel lemma gives the stated result.
\end{proof}

The next lemma will also be frequently used when analyzing protocols: it bounds
the probability of picking a ``bad'' evaluation point.

\begin{lemma}\label{lem:rankeval}
  Let $\matr{A}\in\F[x]^{m\times n}$ with rank at least $r$.
  For any finite subset $\fieldsubset\subseteq\F$ and
  for a point $\alpha\in\fieldsubset$ chosen uniformly at random, the probability
  that $\rank(\matr{A}(\alpha)) < r$ is at most
  $r\deg(\matr{A})/\#\fieldsubset$.
\end{lemma}
\begin{proof}
  Any $r\times r$ minor of $\matr{A}$ has degree at most
  $r\deg(\matr{A})$, and at least one must be nonzero
  since $\rank(\matr{A})\ge r$.
  On the other hand, $\rank(\matr{A}(\alpha)) < r$ if and only if $\alpha$ is a root of
  all such minors.
\end{proof}

\section{Linear algebra operations}\label{sec:vspace}

In this section, we give some verification protocols for the computation of classical linear algebra
properties on polynomial matrices: singularity, rank, and determinant of a matrix,
as well as system solving and matrix multiplication.

The protocols we present here all rely on the same
general idea, which consists in picking a random point and evaluating the input
polynomial matrix (or matrices) at that point.
This allows us to achieve sub-linear communication space.
Note that this technique has been used before by \citet{kns11} to certify
the same properties for integer matrices: in that setup, computations were performed modulo
some prime number, while, in our context, this translates into evaluating polynomials at some element of the base field.

In several of our protocols, the Prover has to solve a linear system over the base field.
For a linear system whose matrix is in \(\matRing{m}{n}\) and has rank $r$,
this can be done in $\tsoh{mnr^{\omega-2}}$ operations in \(\field\)
\citep[see][Algorithm 6]{JPS13}.

\subsection{Singularity and nonsingularity}

We start by certifying the singularity of a matrix. Here, the Verifier picks a random
evaluation point and sends it to the Prover, who evaluates the input matrix at that point
and sends back a nontrivial kernel vector, which the Prover will always
be able to compute since a singular
polynomial matrix is still singular when evaluated at any point. Then, all the Verifier needs
to do is to check that the vector received is indeed a kernel vector. Note that the evaluation
trick here is really what allows us to have a sub-linear --- with respect to the
input size --- communication cost,
as the answer the Prover provides to the challenge is a vector over the base field, and not over the polynomials.

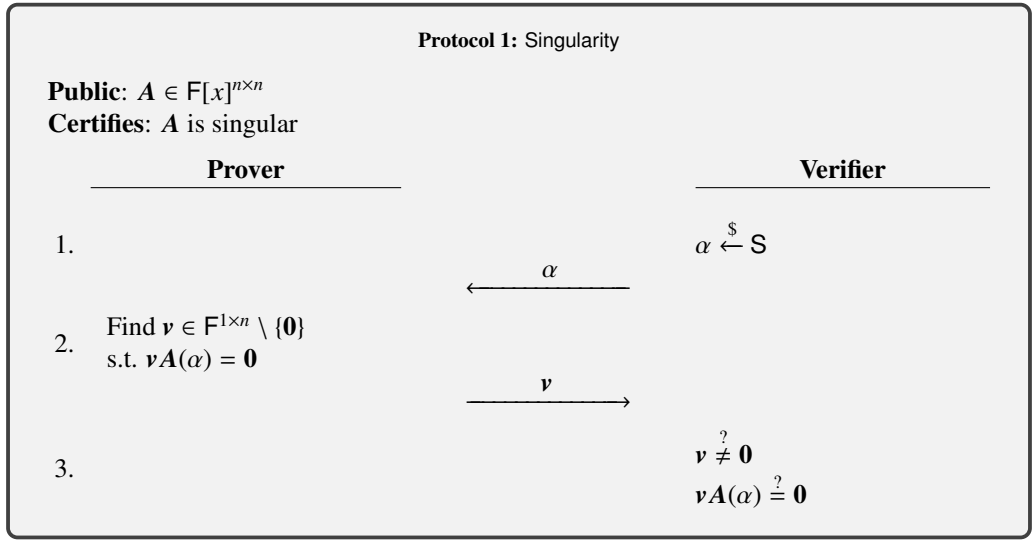
\begin{protocol}
    \caption{\textsf{Singularity}}
    \label{pro:singular}
    \Public{$\matr{A}\in\pmatRing{n}{n}$}
    \Certifies{$\matr{A}$ is singular}
    \begin{protocolsteps}
        \step
            \verifier{$\alpha \randfrom \fieldsubset$}
            \toprover{$\alpha$}
        \step
            \prover{\begin{tabular}{@{}l@{}}
              Find $\vect{v} \in \matRing{1}{n}\setminus\{\zerovect\}$\\
              s.t. $\vect{v}\matr{A}(\alpha) = \zerovect$
            \end{tabular}}
            \toverifier{$\vect{v}$}
        \step
            \verifier{
                $\begin{array}{@{}l@{}}
                    \vect{v} \checkneq \zerovect \\
                    \vect{v} \matr{A}(\alpha) \checkeq \zerovect
                \end{array}$
            }
    \end{protocolsteps}
\end{protocol}

In the next theorem, and for the remainder of the section, for
convenience we write \(d\defeq\max(1,\deg(\matr{A}))\).

\begin{theorem}
    \Cref{pro:singular} is a complete and probabilistically sound interactive
    protocol which requires $\tsoh{n}$ communication and has Verifier cost $\tsoh{n^2 d}$.
    The probability that the Verifier incorrectly accepts is at most $nd/\#\fieldsubset$.
    If $\matr{A}$ is singular, there is an algorithm for the Prover with cost
    $\tsoh{n^2d + nr^{\omega-1}}$.
\end{theorem}

\begin{proof}
    If $\matr{A}$ is singular, $\matr{A}(\alpha)$ must also be singular
    and there exists a nontrivial nullspace vector that the Verifier will
    accept.

    If $\matr{A}$ is nonsingular, then the Prover will be
    able to cheat if the Verifier picked an $\alpha$ such that
    $\matr{A}(\alpha)$ is singular,
    which happens only with probability $nd/\#\fieldsubset$ according to
    \cref{lem:rankeval}.

    Now, for the complexities: the Prover will have to evaluate $\matr{A}$ at $\alpha$,
    which costs $\tsoh{n^2d}$ and to find a nullspace vector over the base field, 
    which costs $\tsoh{nr^{\omega-1}}$, hence the Prover cost.
    The Verifier computes the evaluation and a vector-matrix
    product over $\F$, for a total cost of $\tsoh{n^2 d}$ operations. Finally, a vector over $\F^n$
    and a scalar are communicated, which yields a communication cost of $\tsoh{n}$.
\end{proof}

We now present a verification protocol for nonsingularity. This relies on the same
evaluation-based approach, with one variation: here, we let the Prover provide the evaluation
point. Indeed, if the Verifier picked a random point, they could choose an ``unlucky'' point
for which a nonsingular matrix evaluates to a singular one, and in that
case, the protocol would be incomplete as the Prover will not be able to
convince the Verifier of nonsingularity.
Instead, we let the Prover pick a point
as they have the computational power to find a suitable point
(\cref{step:nonsingular:find_alpha} in \nonsingular{}).
Once this value is committed to the Verifier,
in \crefrange{step:nonsingular_sampleb}{step:nonsingular_checkw}
we use the protocol verifying
nonsingularity over a field due to \citet[Theorem 3]{dk14}.

\begin{protocol}
    \caption{\textsf{NonSingularity}}
  \label{pro:nonsingular}
    \Public{$\matr{A}\in\pmatRing{n}{n}$}
    \Certifies{$\matr{A}$ is nonsingular}
  \begin{protocolsteps}
    \step \label{step:nonsingular:find_alpha}
      \prover{Find $\alpha\in\field$ s.t. \hfill { }\linebreak $\det(\matr{A}(\alpha)) \neq 0$}
      \toverifier{$\alpha$}
    \step\label{step:nonsingular_sampleb}
        \verifier{$\vect{b} \randfrom \fieldsubset^{n \times 1}$}
        \toprover{$\vect{b}$}
    \step
        \prover{Find $\vect{w} \in \F^{n\times 1}$ s.t. \hfill { }\linebreak $\matr{A}(\alpha)\vect{w} =
          \vect{b}$}
        \toverifier{$\vect{w}$}
    \step\label{step:nonsingular_checkw}
      \verifier{$\matr{A}(\alpha)\vect{w} \checkeq \vect{b}$
          }
  \end{protocolsteps}
\end{protocol}
\begin{theorem}\label{thm:nonsing}
    \Cref{pro:nonsingular} is a probabilistically sound interactive protocol
    and is complete assuming that $\#\fieldsubset \geq nd+1.$ It requires
    $\oh{n}$ communication and has Verifier cost $\tsoh{n^2 d}$.  The probability
    that the Verifier incorrectly accepts is at most $1/\#\fieldsubset$. There
    is an algorithm for the Prover with cost $\softoh{n^\omega d}$.
\end{theorem}

\begin{proof}
    If $\matr{A}$ is nonsingular, then, as the field is large enough, there exists
    an $\alpha$ for which the rank of $\matr{A}(\alpha)$ does not drop, and as 
    \crefrange{step:nonsingular_sampleb}{step:nonsingular_checkw} form a 
    complete verification protocol, \nonsingular{} is complete.

    If $\matr{A}$ is singular, it is not possible to find an $\alpha$ such that 
    $\matr{A}(\alpha)$ is nonsingular. This means
    the Prover successfully cheats if they manage to convince the Verifier that 
    $\matr{A}(\alpha)$ is nonsingular, which only happens with probability $1/\#\fieldsubset$ 
     \citep[Theorem 3]{dk14}, hence the soundness of \nonsingular{}.

    Now, for the complexities: the Prover needs to find a suitable $\alpha$.
    The Prover first computes \(\det(\matr{A})\in\polRing\)
    using the deterministic algorithm of \citet[Theorem 1.1]{LaNeZh17}
    in \(\softoh{n^\omega d}\) time.
    Then, using fast multipoint evaluation,
    the determinant is evaluated at \(nd+1\) points from \fieldsubset
    in time
    $\softoh{nd}$ \citep[Corollary 10.8]{vzGG13};
    since \(\deg(\det(\matr{A}))\le nd\), at least one evaluation will
    be nonzero.
    Computing this determinant dominates the later cost for the Prover
    to evaluate \(\matr{A}(\alpha)\) and solve a linear system over the
    base field,
    hence a total cost of $\tssoftoh{n^\omega d}$.

    The Verifier needs to evaluate $\matr{A}$ at $\alpha$ and
    to perform a matrix-vector multiplication over the base field, hence a cost of
    $\tsoh{n^2d}$. Finally, total communications are two vectors of size $n$ over the base
    field and a scalar, hence the cost of $\oh{n}$.
\end{proof}

\subsection{Matrix Rank}

From the protocol for nonsingularity, we immediately infer one for a lower bound \(\rho\)
on the rank: the Prover commits a set of row indices and a set of column indices
which locate a \(\rho\times\rho\) submatrix which is nonsingular,
and then the protocol verifying nonsingularity is run on this submatrix.

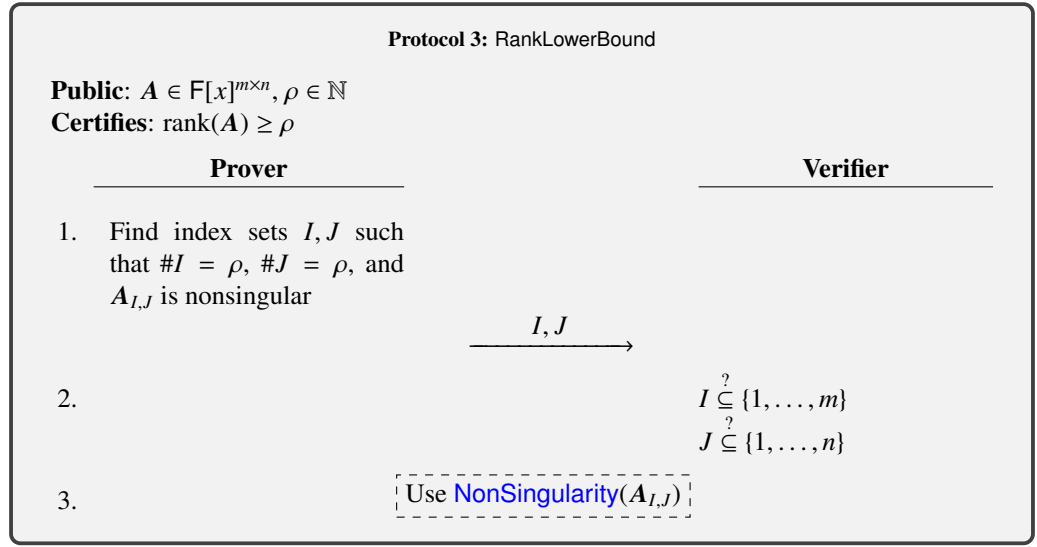
\begin{protocol}
    \caption{\textsf{RankLowerBound}}
    \label{pro:ranklb}
    \Public{$\matr{A}\in\pmatRing{m}{n}$, $\rho \in \NN$}
    \Certifies{$\rank(\matr{A}) \ge \rho$}
    \begin{protocolsteps}
        \step
            \prover{Find index sets \(I,J\) such that
            \(\#I=\rho\), \(\#J=\rho\), and
            \(\matr{A}_{I, J}\) is nonsingular}
            \toverifier{\(I, J\)}
          \step
          \verifier{
          \(I \checksub \{1,\ldots,m\}\) \hfill{ }\linebreak
          \(J \checksub \{1,\ldots,n\} \)
          }
        \vspace{0.2cm} \step
          \subprotocol{Use \nonsingular($\matr{A}_{I, J}$)}
    \end{protocolsteps}
\end{protocol}

\begin{theorem}\label{thm:ranklb}
    Let $r$ be the actual rank of $\matr{A}$. \Cref{pro:ranklb} is a
    probabilistically sound interactive protocol and is complete assuming
    $\#\fieldsubset \geq \rho d+1$ in its subprotocol. It requires $\oh{\rho}$
    communication and has Verifier cost $\tsoh{\rho^2 d}$. If we indeed have
    $r \ge \rho$, then there is a Las Vegas randomized algorithm for the Prover
    with expected cost $\tsoh{mnr^{\omega-2}+mnd}$. Otherwise, the probability
    that the Verifier incorrectly accepts is at most $1/\#\fieldsubset$.
\end{theorem}

\begin{proof}
    If $\rho$ is indeed a lower bound on the rank of $\matr{A}$, 
    there exist two sets $I \subseteq \{1,\ldots,m\}$ and $J \subseteq \{1,\ldots,n\}$ of size \(\rho\) such that 
    $\matr{A}_{I,J}$ is nonsingular, and since
    \nonsingular{}
    is complete, so is this protocol. Note that the completeness of
    the sub-protocol is ensured only if $\#\fieldsubset \ge \rho d+1$.

    If $\rho$ is not a lower bound on the rank of $\matr{A}$, meaning \(\rank(\matr{A})<\rho\),
    then the Prover will not be able to find suitable $I$ and $J$ and hence
    the sets provided by a cheating Prover yield a singular submatrix $\matr{A}_{I,J}$.
    Now, if the Prover provided sets which do not contain $\rho$ elements or which contain elements
    outside the allowed dimension bounds, 
    this will always be detected by the Verifier.
    If the Prover provided sets with enough elements, the Verifier incorrectly accepts
    with the same probability as in \nonsingular{}, which
    is $1/\#\fieldsubset$.

    Regarding the complexities, the Prover has to find a $\rho\times\rho$
    nonsingular submatrix of an $m\times n$ degree $d$ matrix. This can be
    achieved in a Las Vegas fashion, by evaluating the matrix $\matr{A}$ at a
    random $\alpha$ in time $\oh{mnd}$, and computing the rank profile matrix
    (or a rank profile revealing PLUQ decomposition) of $\matr{A}(\alpha)$, see
    for instance \citep{DPS15}. The cost of this computation is
    $\tsoh{mnr^{\omega-2}}$, with $r$ the actual rank of $\matr{A}$.

    As $\rho \leq r$, running \cref{pro:nonsingular} on a $\rho \times \rho$
    matrix does not dominate the complexity, hence the total Prover cost of 
    $\tsoh{mnr^{\omega-2}+mnd}$.
    From \cref{thm:nonsing}, the Verifier cost is
    $\tsoh{\rho^2 d}$. Finally, here two sets of $\rho$ integers are
    transmitted, which with the communications in \nonsingular{} adds up to
    a communication cost of $\oh{\rho}$.
\end{proof}

Now, we give a protocol verifying an upper bound on the rank.
Note that \cref{step:rankub_computegamma,step:rankub_checkgamma} of \cref{pro:rankub}
come from the protocol verifying an upper bound on the rank for matrices over a field
\citep[see][Theorem 4]{dk14}. In this protocol, we use the notation
\(|\cdot|_{H}\) to refer to the Hamming weight: \(| \vect{\gamma}|_{H} \leq
\rho\) means that the vector \(\vect{\gamma}\) as at most \(\rho\) nonzero
entries.

\begin{protocol}
    \caption{\textsf{RankUpperBound}}
    \label{pro:rankub}
    \Public{$\matr{A}\in\pmatRing{m}{n}$, $\rho \in \NN$}
    \Certifies{$\rank(\matr{A}) \le \rho$}
    \begin{protocolsteps}
        \step
            \verifier{
                $\begin{array}{@{}l@{}}
                    \alpha \randfrom \fieldsubset \\
                    \vect{v} \randfrom \fieldsubset^{n \times 1}
                \end{array}$
            }
            \toprover{$\alpha, \vect{v}$}
        \step \label{step:rankub_computegamma}
            \prover{
              Find \(\vect{\gamma}\in\F^{n\times 1}\) such that \(\matr{A}(\alpha)\vect{\gamma} = \matr{A}(\alpha)\vect{v}\)
                    and \(| \vect{\gamma}|_{H} \leq \rho\)
            }
            \toverifier{$\vect{\gamma}$}
        \step \label{step:rankub_checkgamma}
            \verifier{
            $\begin{array}{@{}l@{}}
                |\vect{\gamma}|_{H} \checkle \rho \\
                \matr{A}(\alpha)\vect{\gamma} \checkeq \matr{A}(\alpha)\vect{v}
            \end{array}$
            }
    \end{protocolsteps}
\end{protocol}

\begin{theorem}\label{thm:rankub}
    Let $r$ be the actual rank of $\matr{A}$.  Then, \Cref{pro:rankub} is a
    complete and probabilistically sound interactive protocol which requires
    $\oh{n}$ communication and has Verifier cost $\oh{mnd}$.  If we indeed have
    $r \le \rho$, then there is an algorithm for the Prover with cost bound
    $\tsoh{mnr^{\omega-2} + mnd}$. Otherwise, the probability that the
    Verifier incorrectly accepts is at most $(rd+1)/\#\fieldsubset$.
\end{theorem}

\begin{proof}
    If $\rho$ is indeed an upper bound on the rank of $\matr{A}$, then, 
    whichever evaluation point the Verifier picked, $\rho$ will be an upper
    bound on the rank of $\matr{A}(\alpha)$ and, as the protocol from \citep[Theorem 4]{dk14}
    is complete, this protocol is complete.

    If $\rho$ is not an upper bound on the rank of $\matr{A}$, there are two possibilities
    of failure. Either the Verifier picked an evaluation point for which the rank of 
    $\matr{A}$ drops, which happens with probability at most $rd/\#\fieldsubset$
    by~\Cref{lem:rankeval}; or the Prover managed to cheat during
    the execution of \cref{step:rankub_computegamma,step:rankub_checkgamma} 
    which happens with probability at most $1/\#\fieldsubset$ \citep[Theorem 4]{dk14}.
    Then, the union bound gives a total probability of $(rd+1)/\#\fieldsubset$ for the 
    Verifier to accept a wrong answer.

    The Prover has to evaluate the matrix at $\alpha$ for a cost of $\oh{mnd}$.
    Then to find the vector $\vect{\gamma}$, the Prover can for instance first
    compute a PLUQ decomposition $\matr{A}(\alpha) =
    \matr{P}\begin{bmatrix}\matr{L}_1\\\matr{L}_2\end{bmatrix}
    \begin{bmatrix}\matr{U}_1&\matr{U}_2\end{bmatrix}\matr{Q}$ for a cost of
    $\tsoh{mnr^{\omega-2}}$. Then, the Prover computes the vector
    $\begin{bmatrix}\vect{w}_1\\ \vect{w}_2\end{bmatrix} = \matr{Q}\vect{v}$,
    where $\vect{w}_1$ has size \(r\), and computes $\vect{\gamma} =
    \transpose{\matr{Q}}
    \begin{bmatrix}\vect{w}_1+\matr{U}_1^{-1}\matr{U}_2\vect{w}_2
    \\\vect{0}\end{bmatrix}$. This vector has Hamming weight at most \(r\)
    (recall that \(\matr{Q}\) is a permutation matrix) and satisfies
    $\matr{A}(\alpha)\vect{\gamma} = \matr{P} \begin{bmatrix}
      \matr{L}_1\\\matr{L}_2    \end{bmatrix} \begin{bmatrix}
    \matr{U}_1\vect{w}_1 + \matr{U}_2 \vect{w}_2 \end{bmatrix}
    =\matr{A}(\alpha)\vect{v}$. The Verifier has to evaluate the matrix at
    $\alpha$ and to perform two matrix-vector products over the base field,
    which yields a cost of $\oh{mnd}$. The communication cost is the one of
    sending a scalar and two vectors of size $n$ over the base field, that is,
    $\oh{n}$.
\end{proof}

From these protocols verifying upper bounds and lower bounds on the rank,
we directly obtain one for the rank (\cref{pro:rank}).

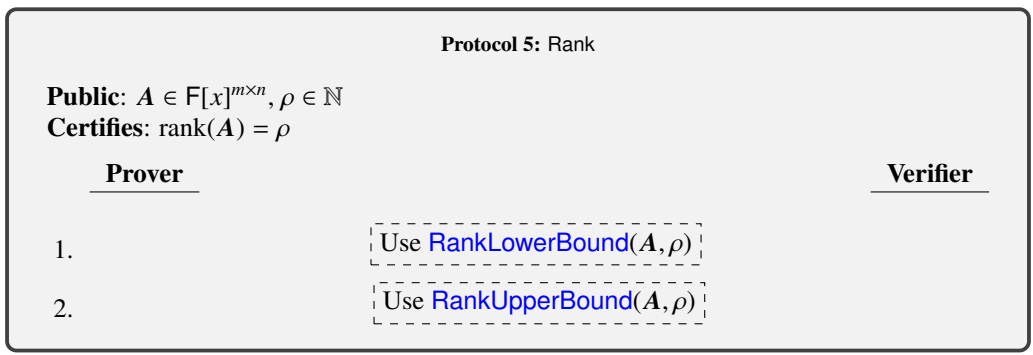
\begin{protocol}
    \caption{\textsf{Rank}}
    \label{pro:rank}
    \Public{$\matr{A}\in\pmatRing{m}{n}$, $\rho \in \NN$}
    \Certifies{$\rank(\matr{A}) = \rho$}
    \begin{protocolsteps}
    \step
        \subprotocol{Use \ranklb($\matr{A}, \rho$)}
    \vspace{0.2cm} \step
        \subprotocol{Use \rankub($\matr{A}, \rho$)}
    \end{protocolsteps}
\end{protocol}

\begin{corollary}
    Let $r$ be the actual rank of $\matr{A}$. \Cref{pro:rank} is a probabilistically sound 
    interactive protocol and is complete assuming $\#\fieldsubset \geq rd+1$ in its 
    subprotocols.
    It requires $\oh{n}$ communication and has Verifier cost $\oh{mnd}$.
    If we indeed have $\rho=r$, then there is a Las Vegas 
    randomized algorithm for the Prover with expected cost 
    $\tsoh{mnr^{\omega-2} + mnd}$. Otherwise,
    the probability that the Verifier incorrectly accepts is at most $(rd+1)/\#\fieldsubset$.
\end{corollary}

\subsection{Determinant}

We follow on with a protocol verifying the determinant of a polynomial matrix,
using a similar evaluation-based approach: after the Verifier has checked that
the degree of the provided determinant is suitable, a random evaluation point
is sampled and the actual verification occurs on the evaluated input.
The check on the degree of the provided determinant is to allow that
the DeMillo-Lipton-Schwartz-Zippel Lemma applies and produces the claimed probability of the 
success.
There are two choices available for the protocol to use over the base field: either
the one from \citep[Section 2]{dktv16}, which runs in a constant number of
rounds but requires a minimum field size of $n^2$, or the one from
\citep[Section 4.1]{dlp17} which runs in $n$ rounds but has no requirement on
the field size. Whichever protocol  is chosen here, this has no impact on the
asymptotic complexities which are the same for both, or on the completeness as
both are perfectly complete.

\begin{protocol}
    \caption{\textsf{Determinant}}
    \label{pro:det}
    \Public{$\matr{A}\in\pmatRing{n}{n}, \delta \in \F[x]$}
    \Certifies{\(\det(\matr{A}) = \delta\)}
    \begin{protocolsteps}
      \step
          \verifier{
          $\begin{array}{@{}l@{}}
              \deg(\delta) \checkle nd \\
              \alpha \randfrom \fieldsubset \\
              \beta \gets \delta(\alpha)
          \end{array}$
          }
          \toprover{$\alpha$}
      \step
          \subprotocol{Use \textsf{FieldDeterminant}($\matr{A}(\alpha), \beta$)}
    \end{protocolsteps}
\end{protocol}

\begin{theorem}
    \Cref{pro:det} is a complete and probabilistically sound interactive
    protocol which requires $\oh{n}$ communication and has Verifier cost
    $\tsoh{n^2 d}$.  If $\delta$ is indeed the determinant of $\matr{A}$, there
    is an algorithm for the Prover which costs $\tsoh{n^2 d + n^\omega}$.
    Otherwise, the probability that the Verifier incorrectly accepts is at most
    $(nd+1)/\#\fieldsubset$.
\end{theorem}

\begin{proof}
    Let \(g=\det(\matr{A})\in\F[x]\) be the actual determinant of \(\matr{A}\).

    If \(\delta=g\), then it must be the case that 
    \(\deg(\delta) \le nd\).
    Then, as \textsf{FieldDeterminant} is complete, the final check will be positive.

    If \(\delta\ne g\), there are two possibilities of failure. If the Verifier
    has picked an $\alpha$ which is a root of \(\delta-g\), then
    \(\delta(\alpha) = g(\alpha)\) and the checks from
    \textsf{FieldDeterminant} will always pass; by the
    DeMillo-Lipton-Schwartz-Zippel lemma, this happens with probability at most
    $nd/\#\fieldsubset$. Otherwise, the Verifier has picked an $\alpha$ such
    that \(\delta(\alpha) \neq g(\alpha)\) which means they will accept
    $\delta$ as the determinant with the probability of failure of
    \textsf{FieldDeterminant}, that is, $1/\#\fieldsubset$.  Overall, by the
    union bound, the probability that the Verifier accepts a wrong statement is
    at most $(nd+1)/\#\fieldsubset$.

    The Prover has to evaluate the matrix at $\alpha$ and to compute a determinant over the
    base field, which yields the cost of $\tsoh{n^2d + n^\omega}$;
    the Verifier has to evaluate $\matr{A}$ at $\alpha$,
    hence a cost of $\tsoh{n^2 d}$; and the communication cost is the one of
    \textsf{FieldDeterminant}, that is, $\oh{n}$.
\end{proof}

\subsection{Protocols based on matrix multiplication}

Finally, we propose verification protocols related to matrix multiplication.
While they are once again based on evaluation techniques, unlike the above
protocols the ones given in this section are non-interactive and thus have no Prover
or communication cost. We first consider linear system solving; recall that
when working over \(\polRing\), given a nonsingular matrix \(\matr{A}\) and a
vector \(\vect{b}\), this problem consists in finding a solution vector
\(\vect{v}\) over \(\polRing\) \emph{together with} a nonzero polynomial
\(\delta\in\polRing\) such that \(\delta^{-1} \vect{v} = \matr{A}^{-1}
\vect{b}\) \citep[see for example][]{GuSaStVa12}.

\begin{protocol}
    \caption{\textsf{SystemSolve}}
    \label{pro:systemsolve}
    \Public{$\matr{A}\in\pmatRing{m}{n}, \vect{b} \in \pmatRing{m}{1}, \vect{v} \in \pmatRing{n}{1}, \delta \in \F[x]$}
    \Certifies{$\deg(\delta) \le \min(m,n) \deg(\matr{A})$ and $\matr{A}\vect{v} = \delta \vect{b}$}
    \begin{protocolsteps}
    \step
        \verifier{
            \hspace*{-2 mm}%
            $\begin{array}{@{}l@{}}
                \alpha \randfrom \fieldsubset \\
                \matr{A}(\alpha)\vect{v}(\alpha) \checkeq
                \delta(\alpha)\vect{b}(\alpha)
            \end{array}$
        }
    \end{protocolsteps}
\end{protocol}

\begin{theorem}\label{th:systemsolve}
    Let $d$ be an upper bound on the degree of \matr{A}, \vect{v},
    \vect{b}, and $\delta$.
    Then, \Cref{pro:systemsolve} is a complete and probabilistically sound non-interactive
    protocol which has Verifier cost $\oh{mnd}$.
    The probability that the Verifier incorrectly accepts is at most $2d/\#\fieldsubset$.
\end{theorem}

\begin{proof}
    If \(\matr{A}\vect{v}=\delta\vect{b}\), then the same holds when evaluating
    at \(\alpha\), hence the completeness of this protocol.

    Otherwise, we have $\matr{A}\vect{v} -\delta \vect{b} = \matr{\Delta}$ for
    some nonzero vector $\matr{\Delta}\in\pmatRing{m}{1}$, and the Verifier
    incorrectly accepts if and only if the Verifier picked an $\alpha$ such
    that \(\matr{\Delta}(\alpha)=0\). Using \cref{lem:rankeval} with the
    vector \(\matr{\Delta}\) of rank \(1\) and degree at most \(2d\), it
    follows that the Verifier picked such an \(\alpha\) with probability at
    most \(2d / \#\fieldsubset\).

    The dominating step in the Verifier's work is evaluating $\matr{A}$ at
    $\alpha$, which costs $\oh{mnd}$ operations in \(\field\).
\end{proof}

Remark that when solving linear systems over \(\polRing\) one is often
interested in the case of a nonsingular matrix \(\matr{A}\) with \(m = n\). In
this context, one usually seeks a solution \((\vect{v},\delta)\) with $\delta$
of minimal degree (see \citep[Section 9]{Storjohann03} and \citep[Section
7]{GuSaStVa12}); this implies \(\deg(\delta) \le \deg(\det(\matr{A})) \le n
\deg(\matr{A})\) and \(\deg(\vect{v}) = \deg(\delta \matr{A}^{-1} \vect{b}) \le
(n-1) \deg(\matr{A}) + \deg(\vect{b})\). In this particular case, these bounds
could be checked by the Verifier at the beginning of the protocol; the
probability that the Verifier incorrectly accepts becomes $(nd_{\matr{A}} +
d_{\vect{b}}) / \#\fieldsubset$; and the Verifier's work costs \(\tsoh{n^2
d_{\matr{A}} + n d_{\vect{b}}}\) operations.

Similarly, we propose a protocol verifying matrix multiplication
following an approach attributed to \citet{Fre79}.

\begin{protocol}
    \caption{\textsf{MatMul}}
    \label{pro:matmul}
    \Public{\(\matr{A}\in\pmatRing{m}{n}, \matr{B}\in\pmatRing{n}{\ell}, \matr{C}\in\pmatRing{m}{\ell}\)}
    \Certifies{\(\matr{C} = \matr{A}\matr{B}\)}
    \begin{protocolsteps}
        \step
            \verifier{
                \hspace*{-4 mm}%
                $\begin{array}{@{}l@{}}
                    \deg(\matr{C}) \stackrel{?}{\le}
                      \deg(\matr{A}) + \deg(\matr{B}) \\
                    \alpha \randfrom \fieldsubset \\
                    \vect{v} \randfrom \fieldsubset^{\ell \times 1} \\
                    \matr{C}(\alpha)\vect{v}
                      \checkeq \matr{A}(\alpha)(\matr{B}(\alpha)\vect{v})
                \end{array}$
            }
    \end{protocolsteps}
\end{protocol}

\begin{theorem}\label{thm:polyfreivalds}
    Let $d_{\matr{A}} = \max(1,\deg(\matr{A}))$ and similarly for
    $d_{\matr{B}}, d_{\matr{C}}$.
    \Cref{pro:matmul} is a complete and probabilistically sound non-interactive
    protocol which has Verifier cost
    $\oh{mnd_{\matr{A}} + n\ell d_{\matr{B}} + m\ell d_{\matr{C}}}$.
    The probability that the Verifier incorrectly accepts is at most
    $(d_{\matr{A}}+d_{\matr{B}}+1)/\#\fieldsubset$.
\end{theorem}

\begin{proof}
    Let $\matr{D}$ be the actual product $\matr{D} = \matr{A}\matr{B}$, and let
    $\matr{\Delta} = \matr{D} - \matr{C}$.
    Note that the final
    check of the Verifier is equivalent to \(\matr{\Delta}(\alpha)\vect{v} \checkeq \matr{0}\).

    If $\matr{C} = \matr{D}$, then $\matr{\Delta} = \matr{0}$ and 
    whichever evaluation point $\alpha$ the Verifier picked,
    we have $\matr{\Delta}(\alpha) = \matr{0}$. The degree bound
    checked initially by the Verifier is also valid whenever
    \(\matr{A}\matr{B}=\matr{C}\), hence this
    protocol is complete.

    Otherwise, $\matr{\Delta}$ is a nonzero matrix
    with degree at most \(d_{\matr{A}}+d_{\matr{B}}\).
    There are two events that would lead to the Verifier accepting incorrectly.
    First, if the Verifier picked an evaluation point such that
    $\matr{\Delta}(\alpha)=\matr{0}$, which happens with probability
    at most $(d_{\matr{A}}+d_{\matr{B}})/\#\fieldsubset$ by \cref{lem:rankeval} (with rank lower bound $1$),
    then whichever verification vector $\vect{v}$ is picked afterwards, the Verifier will always
    accept. Otherwise, the Verifier picked an evaluation point for which
    \(\matr{\Delta}(\alpha)\neq\matr{0}\) but
    they picked a unlucky verification vector $\vect{v}$, that is, \(\vect{v}\) is in the right
    kernel of \(\matr{\Delta}(\alpha)\), which happens
    with probability at most $1/\#\fieldsubset$ according to \cref{lem:randinnerprod}.
    The union bound gives the stated bound for the probability
    that the Verifier incorrectly accepts.
    
    The cost for the Verifier comes from evaluating all three matrices
    at $\alpha$ and then performing three
    matrix-vector products over \(\field\).
\end{proof}

Verifying a matrix inverse is a straightforward application of the previous
protocol.

\begin{corollary}
  \label{cor:inversion}
  For $\matr{A} \in \pmatRing{n}{n}$ and $\matr{B} \in \pmatRing{n}{n}$, there
  exists a non-interactive protocol which certifies that $\matr{B}$ is the
  inverse of $\matr{A}$ in Verifier cost $\tsoh{n^2d}$, where \(d =
  \max(1,\deg(\matr{A}),\deg(\matr{B}))\). If \(\matr{B}\ne\matr{A}^{-1}\), the
  probability that the Verifier incorrectly accepts is at most $(2d +
  1)/\#\fieldsubset$.
\end{corollary}

\section{Row space membership}\label{sec:rowmem}

In this section we present the main tool for verification problems that
are essentially about $\F[x]$-modules, which is to determine whether a
given row vector $\vect{v}\in\F[x]^{1\times n}$ is in the $\F[x]$-row space of
a given matrix $\matr{A}\in\F[x]^{m\times n}$.

The approach has two steps. First, \rowmemf{} shows how to
solve the problem in case $\matr{A}$ has full row rank. Then, in
\rowmem{}, we extend this to the general setting by designing
a reduction to several calls to the full row rank case.

\subsection{Full row rank case}

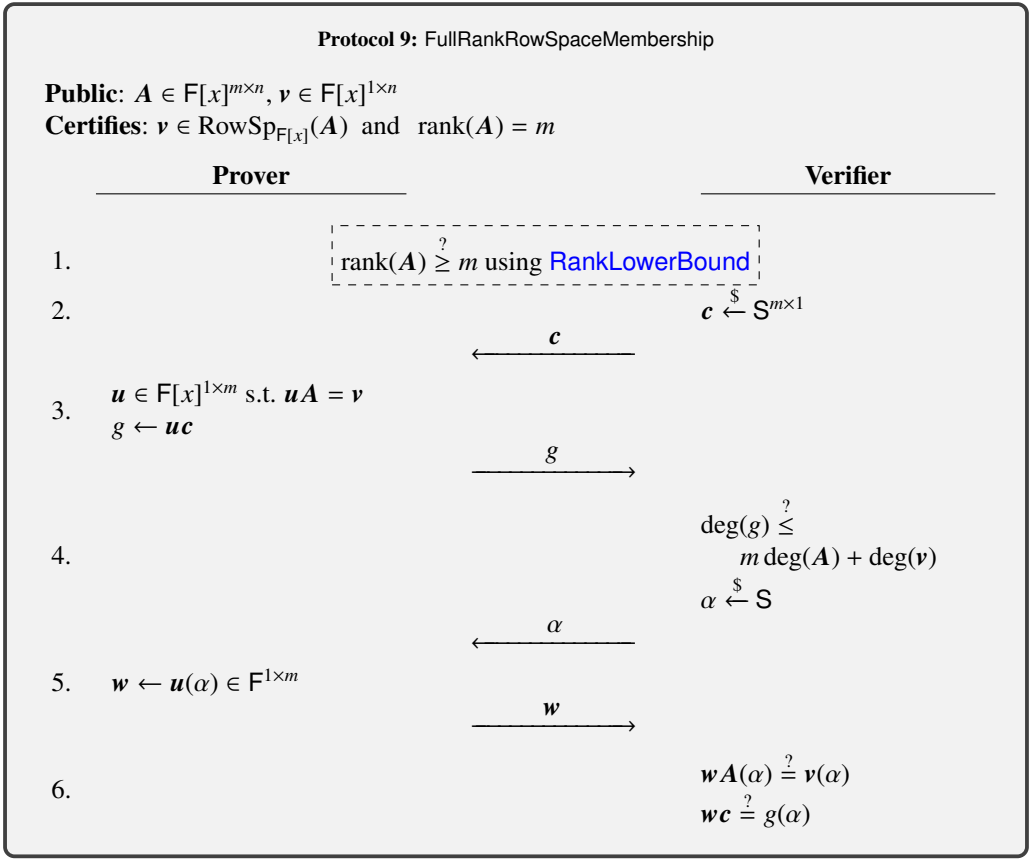
\begin{protocol}
  \caption{\textsf{FullRankRowSpaceMembership}}
  \label{pro:rowmemf}
  \Public{$\matr{A}\in\F[x]^{m\times n}$, $\vect{v}\in\F[x]^{1\times n}$}
  \Certifies{$\vect{v} \in \rowsp_{\F[x]}(\matr{A}) \enspace\text{and}\enspace \rank(\matr{A}) = m$}
  \begin{protocolsteps}
  \step\label{rowmemf:rank}
    \subprotocol{\(\rank(\matr{A}) \checkge m\) using \ranklb{}}
  \step \label{step:rowmem_choosec}
    \verifier{$\vect{c} \randfrom \fieldsubset^{m\times 1}$}
  \toprover{$\vect c$}
  \step \label{step:rowmem_computeug}
    \prover{%
      \(\begin{array}{@{}l@{}}
        \vect{u}\in\F[x]^{1\times m} \text{ s.t. } \vect{u}\matr{A}=\vect{v} \\
       g \gets \vect{u}\vect{c}
     \end{array}\)
    } %
  \toverifier{$g$}
  \step \label{step:rowmem:checkg}
    \verifier{%
    $\begin{array}{@{}l@{}}
       \deg(g) \checkle \\
       \phantom{deg} m\deg(\matr{A}) + \deg(\vect{v}) \\
       \alpha \randfrom \fieldsubset
     \end{array}$
     } %
  \toprover{$\alpha$}
  \step \label{step:rowmem_sendw}
    \prover{%
      $\vect{w} \gets \vect{u}(\alpha) \in \matRing{1}{m}$}
  \toverifier{$\vect{w}$}
  \step \label{step:rowmem_verify}
    \verifier{%
      $\begin{array}{@{}l@{}}
         \vect{w} \matr{A}(\alpha) \checkeq \vect{v}(\alpha) \\
         \vect{w} \vect{c} \checkeq g(\alpha)
       \end{array}$
     } %
  \end{protocolsteps}
\end{protocol}

For this case, we propose \cref{pro:rowmemf}. Before studying its properties,
we emphasize that its soundness crucially depends on the fact that \(\matr{A}\)
has full row rank.
To see why, let \(\matr{A} =
\transpose{[\var \;\; -\var]}\) and \(\vect{v} = [1]\), and write \(\vect{c} =
\transpose{[c_1 \;\; c_2]}\) for the random vector chosen by the Verifier.
Here, \(\matr{A}\) does not have full row rank and \(\vect{v}\) is not in the
row space of \(\matr{A}\); it is however in the rational row space of
\(\matr{A}\). This allows a dishonest Prover to make the Verifier accept by
means of forging a \emph{rational} vector \(\vect{u}\) such that
\(\vect{u}\matr{A} = \vect{v}\) and \(\vect{u}\vect{c}\in\polRing\): the
Verifier cannot detect that \(\vect{u}\) was not over \(\polRing\), since they
only receive \(\vect{u}\vect{c}\) and an evaluation of \(\vect{u}\). Indeed,
any vector of the form \(\vect{u}= [\var^{-1}+f(\var) \;\; f(\var)]\) for some
\(f\in\fracRing\) is such that \(\vect{u}\matr{A} = \vect{v}\). In the likely
event that \(c_1+c_2 \neq 0\), the Prover can choose any polynomial
\(g\in\polRing\) and define \(f = (c_1+c_2)^{-1} (g - c_1\var^{-1})\); then
\(\vect{u}\vect{c}=g\) is a polynomial.

Remark that if \(\matr{A}\) has full row rank and \(\vect{v}\) belongs to the
rational row space of \(\matr{A}\), then we have \emph{uniqueness} of the
(rational) vector \(\vect{u}\) such that \(\vect{u}\matr{A}=\vect{v}\) and thus
there is no flexibility for the Prover on the choice of \(\vect{u}\). In this
case, the following lemma plays a key role in the soundness of
\cref{pro:rowmemf}.

\begin{lemma}\label{lem:ratinnerprod}
  Let $\vect{u}\in\F(x)^{1\times n}$ be a rational fraction vector with
  $\denom(\vect{u}) \ne 1$ and let $\fieldsubset \subseteq \F$ be a finite subset.
  For a vector of scalars $\vect{c}\in\fieldsubset^{n\times 1}$
  chosen uniformly at random,
  the probability that the inner product \(\vect{u}\vect{c}\) is a polynomial,
  i.e., that $\denom(\vect{u}\vect{c}) = 1$, is at most
  $1/\#\fieldsubset$.
\end{lemma}
\begin{proof}
  Write $f=\denom(\vect{u})$ and $\hat{\vect{u}} = \numer(\vect{u})$.
  By the condition of the lemma we know that
  $\deg(f) \ge 1$.
  We see that the inner product of $\vect{u}$ and $\vect{c}$ is a
  polynomial if and only if the inner product of $\hat{\vect{u}}$ and
  $\vect{c}$ is divisible by $f$.

  Now let $h$ be any irreducible factor of $f$, and consider the
  inner product \(\hat{\vect{u}} \vect{c}\) with
  \(\hat{\vect{u}}\) seen as a vector over the extension field $\F[x]/\langle h\rangle$.
  Because $h \mid \denom(\vect{u})$,
  we know that $\hat{\vect{u}} \bmod h$ is not zero; otherwise the
  degree of the denominator $f$ is not minimal.
  Then, since $\fieldsubset \subseteq \F \subseteq \F[x]/\langle h\rangle$, the
  stated bound follows from \cref{lem:randinnerprod}.
\end{proof}

Another ingredient for our full row rank space membership protocol is a
subroutine the Prover may use to actually compute the solution \(\vect{u}\) to
the linear system, shown in \cref{alg:linsolve}. More precisely, this algorithm
computes the numerator \(\hat{\vect{u}}\) and the corresponding minimal
denominator \(f\). This algorithm will also be used in the protocol for
arbitrary-rank matrices presented in the next section.

\begin{algorithm}
  \caption{Linear system solving with full row rank}
  \label{alg:linsolve}
  \KwIn{\(\matr{A}\in\F[x]^{m\times n}, \vect{v}\in\F[x]^{1\times n}\)}
  \KwOut{Either \lowrank{}, or \nosol{}, or
    \((\hat{\vect{u}},f) \in (\polRing^{1\times m}\times \polRing)\) such that
    \(\hat{\vect{u}}\matr{A}=f\vect{v}\) and \(f\) has minimal degree}
    \(r, i_1,\ldots,i_r \gets\) column rank profile of \(\matr{A}\)
    \label{step:linsolve:crp}\;
    \lIf{\(r < m\)}{\KwRet{\lowrank{}} \hfill \texttt{// below, \(r=m\)}}
    \(\matr{B} \in \pmatRing{r}{r}\gets\) columns \(i_1,\ldots,i_r\) from \(\matr{A}\)\;
    \(\vect{y} \in \pmatRing{1}{r} \gets\) columns \(i_1,\ldots,i_r\) from \(\vect{v}\) \;
    Compute \((\hat{\vect{u}},f) \in (\polRing^{1\times m} \times \polRing)\) such that
    \(f^{-1}\hat{\vect{u}} = \vect{y}\matr{B}^{-1}\) and \(f\) has minimal
    degree, using \citep[Algorithm \texttt{RationalSystemSolve}]{GuSaStVa12}
  \label{step:linsolve:linsolve}\;
  \lIf{\(\hat{\vect{u}}\matr{A} \ne f\vect{v}\)}{\KwRet{\nosol{}}}
  \KwRet{\(\hat{\vect{u}}\)}
\end{algorithm}

As above, to simplify the cost bounds we write \(d_{\matr{A}} \defeq
\max(1,\deg(\matr{A}))\) and \(d_{\vect{v}} \defeq \max(1,\deg(\vect{v}))\).

\begin{lemma}\label{lem:linsolve}
  \Cref{alg:linsolve} uses \(\tssoftoh{m^{\omega-1}nd_{\matr{A}} +
  m^{\omega-1}d_{\vect{v}}}\) operations in \(\field\).  If \(\matr{A}\) has
  rank less than \(m\), then \lowrank{} is returned.  If \(\matr{A}\) has rank
  \(m\) and \(\vect{v}\not\in\rowsp_{\F(x)}(\matr{A})\), then \(\nosol\) is
  returned.  Otherwise, the algorithm returns \((\hat{\vect{u}},f)\) such that
  \(\hat{\vect{u}}\matr{A}=f\vect{v}\) and \(f\) has minimal degree; in
  particular, \(\deg(f) \le m \deg(\matr{A})\) and \(\deg(\hat{\vect{u}}) \le
  (m-1)\deg(\matr{A}) + \deg(\vect{v})\).
\end{lemma}
\begin{proof}
  \citep[Chapter 11]{Zhou12} presents a deterministic algorithm to compute the
  column rank profile on \cref{step:linsolve:crp} using
  \(\tssoftoh{m^{\omega-1}nd_{\matr{A}}}\) field operations.
  This guarantees that \lowrank{} is returned whenever \(\matr{A}\) does not
  have full row rank.

  Now assume that \(\rank(\matr{A})=m\). Then \(\matr{B}\) is nonsingular, and
  \citet{GuSaStVa12} showed how to solve the linear system on
  \cref{step:linsolve:linsolve} deterministically using \(\tssoftoh{m^\omega
  d_{\matr{A}} + m^{\omega-1} d_{\vect{v}}}\) operations; precisely, this cost
  bound is obtained from the results in \citepalias[Section 7]{GuSaStVa12}
  applied with \(d = \max(d_{\matr{A}}, m d_{\vect{v}})\). The degree bounds on
  \(\hat{\vect{u}}\) and \(f\) follow from Cramer's rule.
  
  Let \((\hat{\vect{u}},f)\) be the system solution computed on
  \cref{step:linsolve:linsolve}. If
  \(\vect{v}\not\in\rowsp_{\F(x)}(\matr{A})\), then we must have
  \(\hat{\vect{u}}\matr{A}\neq f\vect{v}\) and thus \(\nosol\) is returned. Now
  assume that \(\vect{v}\in\rowsp_{\F(x)}(\matr{A})\), that is, there exists
  \(\vect{w}\in\F(x)^{1\times m}\) such that \(\vect{w}\matr{A}=\vect{v}\).
  Then we have in particular \(\vect{w}\matr{B}=\vect{y}\). But because
  \(\matr{B}\) is nonsingular, we have \(\vect{w} = \vect{y}\matr{B}^{-1} =
  f^{-1}\hat{\vect{u}}\); hence \(\hat{\vect{u}}\matr{A}=f\vect{v}\).
\end{proof}

Finally, we present the main result of this subsection.

\begin{theorem}\label{thm:rowmemf}
  \Cref{pro:rowmemf} is a complete and probabilistically sound interactive
  protocol which requires \(\oh{md_{\matr{A}}+d_{\vect{v}}}\) communication and
  has Verifier cost \(\oh{mnd_{\matr{A}} + nd_{\vect{v}}}\). If
  \(\rank(\matr{A})=m\) and \(\vect{v}\in\rowsp_{\F[x]}(\matr{A})\), there is
  an algorithm for the Prover with cost \(\tssoftoh{nm^{\omega-1}d_{\matr{A}} +
  m^{\omega-1}d_{\vect{v}}}\). Otherwise, the probability that the Verifier
  incorrectly accepts is at most \((3md_{\matr{A}} + d_{\vect{v}} +
  1)/\#\fieldsubset\).
\end{theorem}
\begin{proof}
  If \(\rank(\matr{A}) < m\), then from \cref{thm:ranklb}, the probability
  that the Verifier incorrectly accepts is at most \(1/\#\fieldsubset\), less
  than the stated bound in the theorem.
  And if $\vect{v}$ is the zero vector, then the protocol easily succeeds
  when the Prover sends all zeros for $g$ and $\vect{w}$; remark that
  \(\vect{u}=\zerovect\) is the only solution to \(\vect{u}\matr{A}=\vect{v}\)
  when \(\matr{A}\) has full row rank.
  So for the remainder of the proof, assume that $\vect{v}$ is nonzero
  and $\matr{A}$ has full row rank $m$.

  The rank check entails \(2m+1\) field elements of communication, and
  the degree check by the Verifier assures that $g$ contains at most
  $md_{\matr{A}} + d_{\vect{v}} + 1$ field elements, bringing the total communication
  in the protocol
  to at most $m(d_{\matr{A}}+4) + d_{\vect{v}} + 3$ field elements.

  The work of the Verifier is dominated by computing the evaluations
  $\matr{A}(\alpha)$ and $\vect{v}(\alpha)$ on the
  last step. Using Horner's method the total cost for these is
  $\oh{mnd_{\matr{A}} + nd_{\vect{v}}}$, as claimed.

  We now divide the proof into three cases, depending on whether $\vect{v}$
  is in the polynomial row space of $\matr{A}$ (as checked by the
  protocol), the rational row space of $\matr{A}$, or neither.

  \paragraph{Case 1: $\vect{v}\in\rowsp_{\F[x]}(\matr{A})$}

  Here we want to prove that an honest Prover and Verifier succeed with
  costs as stated in the theorem.

  The vector $\vect{u}$ as defined in \cref{step:rowmem_computeug}
  must exist by the definition of $\rowsp_{\F[x]}$,
  and computing $\vect{u}$ can be completed by the Verifier according to
  \cref{lem:linsolve} in the stated cost bound.

  If the computations of $\vect{u}$ and $g$ at \cref{step:rowmem_computeug} and
  of \(\vect{w}\) at \cref{step:rowmem_sendw} are performed correctly by the
  Prover, then the Verifier's checks on \cref{step:rowmem_verify} will succeed
  for any choice of $\alpha$. Note also that in this case, \(\deg(g) =
  \deg(\vect{u}\vect{c}) \le \deg(\vect{u})\), and \(\deg(\vect{u}) \le
  m\deg(\matr{A}) + \deg(\vect{v})\) holds (see \cref{lem:linsolve}),
  hence the degree check at \cref{step:rowmem:checkg}.

  This proves the completeness of the protocol.

  \paragraph{Case 2: $\vect{v}\in\rowsp_{\F(x)}(\matr{A}) \setminus
    \rowsp_{\F[x]}(\matr{A})$}

  In this case, the assertion being verified is false, and we want to
  show probabilistic soundness.

  Let $\vect{c}\in\F^{m\times 1}$ be the random vector chosen by the Verifier on
  \cref{step:rowmem_choosec}.
  Since \(\matr{A}\) has full row rank, there is a unique rational solution
  $\vect{u}\in\F(x)^{1\times m}$ such that $\vect{u}\matr{A}=\vect{v}$, and
  by the assumption of this case we have \(\denom(\vect{u}) \neq 1\); besides,
  \cref{lem:linsolve} ensures \(\deg(\denom(\vect{u})) \le md_{\matr{A}}\)
  and \(\deg(\numer(\vect{u})) \le (m-1)d_{\matr{A}} + d_{\vect{v}}\).
  Then, \cref{lem:ratinnerprod} tells us that the probability that
  $\vect{u}\vect{c}$ is a polynomial is at most $1/\#\fieldsubset$.
  Let \(g\) be the polynomial sent by the Prover at \cref{step:rowmem_computeug}.
  If $\vect{u}\vect{c}$ is not a polynomial, then
  $\vect{u}\vect{c}-g$ is a nonzero rational fraction with numerator
  degree at most
  \begin{equation}\label{eqn:rowmem:degbound}
    \max(\deg(\numer(\vect{u})), \deg(g) + \deg(\denom(\vect{u}))) \le 2md_{\matr{A}} + d_{\vect{v}}.
  \end{equation}

  From \cref{lem:rankeval}, the probability that
  $\matr{A}(\alpha)$ does not have full row rank is at most $md_{\matr{A}}/\#\fieldsubset$.
  Otherwise, the vector $\vect{w}=\vect{u}(\alpha)$ is the
  unique solution to $\vect{w}\matr{A}(\alpha)=\vect{v}(\alpha)$, so the
  Prover is obliged to send this $\vect{w}$ on
  \cref{step:rowmem_sendw}.

  Then, if the Verifier incorrectly accepts, we must have
  $\vect{w}\vect{c} = g(\alpha)$, which means
  $\vect{u}(\alpha)\vect{c} = g(\alpha)$, or in other words,
  \(\alpha\) is a root of \(\vect{u}\vect{c} - g\).
  The degree bound in \cref{eqn:rowmem:degbound} gives an upper
  bound on the number of such roots $\alpha\in\F$.

  Therefore the Verifier accepts only when either
  $\vect{u}\vect{c}\in\F[x]$, or $\matr{A}(\alpha)$ is
  singular, or $\alpha$ is a root of $\vect{u}\vect{c}-g$, which by the
  union bound has probability at most
  $(3md_{\matr{A}}+d_{\vect{v}}+1)/\#\fieldsubset$, as stated.

  \paragraph{Case 3: $\vect{v}\notin\rowsp_{\F(x)}(\matr{A})$}

  Again, the assertion being verified is false, and our goal is to
  prove probabilistic soundness.
  As with the last case, assume by way
  of contradiction that the Verifier accepts.

  Consider the augmented matrix
  \[
    \tilde{\matr{A}} \defeq
    \begin{bmatrix}
      \matr{A} \\ \vect{v}
    \end{bmatrix} \in \pmatRing{(m+1)}{n}.
  \]
  By the assumption of this case, $\rank(\tilde{\matr{A}}) =
  \rank(\matr{A}) + 1 = m+1$.
  But the vector $\vect{w}$ provided at \cref{step:rowmem_verify} is such that
  $\vect{w}\matr{A}(\alpha) = \vect{v}(\alpha)$: it
  corresponds to a nonzero vector \([-\vect{w} \;\; 1]\) in the left kernel of
  $\tilde{\matr{A}}(\alpha)$, which therefore has rank at most $m$.

  Since all \((m+1)\times(m+1)\) minors of \(\tilde{\matr{A}}\) have
  degree at most \(md_{\matr{A}} + d_{\vect{v}}\),
  the proof of \cref{lem:rankeval} shows that the probability that
  $\rank(\tilde{\matr{A}}(\alpha)) \le m$ is at most
  $(md_{\matr{A}} + d_{\vect{v}})/\#\fieldsubset$.
\end{proof}

\subsection{Arbitrary rank case}

Now we move to the general case of a matrix $\matr{A}$ with arbitrary rank \(r\).

The idea behind our protocol is inspired by \citet{MS04}.
We make use of the full row rank case by
considering a matrix \(\matr{C}\in\F[x]^{r\times m}\) such that
\(\matr{C}\matr{A}\) has full row rank. Thus \(\matr{C}\matr{A}\) has the same
rational row space as \(\matr{A}\),
and if \(\vect{v}\in\rowsp_{\F[x]}(\matr{A})\), then there is a unique rational vector \(\vect{w}\in\F(x)^{1\times r}\) such
that \(\vect{w}\matr{C}\matr{A}=\vect{v}\). In particular, for \(\dd =
\denom(\vect{w})\) we have \(\dd \vect{v} \in \rowsp_{\polRing}(\matr{A})\), and
therefore if \(\dd=1\) the verification is already complete.

Although it might be that \(\vect{w}\) has a nontrivial denominator \(\dd\),
this approach can still be used for verification by considering \emph{several}
such matrices \(\matr{C}_1,\ldots,\matr{C}_t\) and rational vectors
\(\vect{w}_1,\ldots,\vect{w}_t\) with denominators \(\dd_1,\ldots,\dd_t\).
Indeed, we will see that these matrices can be chosen such that the greatest
common divisor of \(\dd_1,\ldots,\dd_t\) is \(1\); as we show in the next
lemma, this implies \(\vect{v}\in\rowsp_{\F[x]}(\matr{A})\).

\begin{lemma}
  \label{lem:denom_gcd}
  Let \(\matr{A}\in\pmatRing{m}{n}\) and \(\vect{v}\in\pmatRing{1}{n}\). Let
  $\dd_1,\ldots,\dd_t\in\F[x]$ be such that
  \(\dd_i\vect{v}\in\rowsp_{\polRing}(\matr{A})\) for \(1 \le i \le t\). If
  \(\gcd(\dd_1,\ldots,\dd_t)=1\), then \(\vect{v}\in\rowsp_{\polRing}(\matr{A})\).
\end{lemma}
\begin{proof}
  The gcd assumption implies that there exist \(u_1,\ldots,u_t\in\polRing\)
  such that \(u_1\dd_1 + \cdots + u_t\dd_t = 1\). It directly follows that
  \(\vect{v} = u_1 (\dd_1\vect{v}) + \cdots + u_t (\dd_t \vect{v})\) belongs to
  \(\rowsp_{\polRing}(\matr{A})\).
\end{proof}

Before giving the full protocol for row membership, we first present a
subprotocol \progcd{} to confirm that the greatest common divisor of
a set of polynomials is 1.

\begin{protocol}
  \caption{\textsf{CoPrime}}
  \label{pro:gcd}
  \Public{\(t\ge 2\) polynomials \(\dd_1,\ldots,\dd_t\in\F[x]\)}
  \Certifies{\(\gcd(\dd_1,\ldots,\dd_t) = 1\)}
  \begin{protocolsteps}
  \step
    \prover{\(\begin{array}{@{}l@{}}
      \text{Compute polynomials } s_1,s_2\in\F[x]\\
      \text{and scalars } \beta_3,\ldots,\beta_t\in\F\\
      \text{s.t. } \dd_1s_1 + hs_2 = 1,\\
      \deg(s_1) < \deg(h), \\
      \text{and } \deg(s_2) < \deg(\dd_1), \\
      \text{where } h \defeq \dd_2 + \sum_{i=3}^t \beta_i\dd_i
    \end{array}\)}
    \toverifier{\(s_1,s_2,\beta_3,\ldots,\beta_t\)}
  \step
    \verifier{\hspace*{-1em}\(\begin{array}{@{}l@{}}
      \deg(s_1) \checklt \max_{i\ge 2}\deg(\dd_i) \\
      \deg(s_2) \checklt \deg(\dd_1) \\
      \alpha \randfrom \fieldsubset \\
      h_\alpha \gets \dd_2(\alpha) + \sum_{i=3}^t \beta_i\dd_i(\alpha) \\
      \dd_1(\alpha)s_1(\alpha) + h_\alpha s_2(\alpha) \checkeq 1
    \end{array}\)}
  \end{protocolsteps}
\end{protocol}

\begin{lemma}\label{lem:gcd}
  Let \(d=\max_i\deg(\dd_i)\), and suppose \(\#\fieldsubset\ge 2d\).
  Then \cref{pro:gcd} is a complete and probabilistically sound interactive
  protocol which requires \(\oh{d+t}\) communication and has Verifier
  cost \(\oh{dt}\). If \(\gcd(\dd_1,\ldots,\dd_t) = 1\), then there is a Las
  Vegas randomized algorithm for the Prover with expected cost bound
  \(\softoh{dt}\). Otherwise, the probability that the Verifier incorrectly
  accepts is at most \((2d-1)/\#\fieldsubset\).
\end{lemma}
\begin{proof}
  The communication and Verifier costs are clear.

  Write \(g \defeq \gcd(\dd_1,\ldots,\dd_t)\), and suppose first that
  \(g\ne 1\). Since \(g\) divides \(\dd_1s_1+hs_2\), the polynomial
  \(\dd_1s_1+hs_2-1\) is nonzero and has degree at most \(2d-1\).
  If the Verifier incorrectly accepts, then $\alpha$ must be a root of
  this polynomial, which justifies the probability claim.

  If \(g=1\), then a well-known argument
  \citep[Theorem 6.46]{vzGG13}
  says that, for \(\beta_3,\ldots,\beta_t\) chosen randomly from a
  subset \(\fieldsubset\subseteq\F\), the probability that
  \(\gcd(\dd_1,h)\ne \gcd(\dd_1,\ldots,\dd_t)\) is at most
  \(d/\#\fieldsubset\). Based on the assumption that
  \(\#\fieldsubset\ge 2d\), the Prover can find such a tuple
  \(\beta_3,\ldots,\beta_t\) after expected \(O(1)\) iterations. Then
  computing the B\'ezout coefficients \(s_1,s_2\) is done via
  the fast extended Euclidean algorithm on \(\dd_1\) and \(h\),
  which costs \(\softoh{d t}\).
\end{proof}

\Cref{pro:rowmem} shows an interactive protocol verifying row space membership.
For free (and as a necessary aspect of the protocol), the rank \(\rho\) is
also verified.

The Prover first selects \(t\)
matrices \(\matr{C}_i\in\field^{r \times m}\)
such that \(\matr{C}_i\matr{A}\) has full row rank \(r=\rank(\matr{A})\) and
the corresponding denominators \(\dd_i\) of the
rational solutions to \(\vect{w}\matr{C}_i\matr{A}=\vect{v}\) have no
common factor.

The Verifier then confirms that the gcd of all
denominators is 1 using \progcd{}.
Using \rowmemf{}, the Verifier also
confirms that each \(\matr{C}_i\matr{A}\) has full rank and
each \(\dd_i \vect{v}\) is in the row space of \(\matr{C}_i\matr{A}\)
and therefore in the row space of \(\matr{A}\) as well; by \cref{lem:denom_gcd} this ensures
that \(\vect{v}\) is itself in the row space of \(\matr{A}\).

To save communication costs, the matrices \(\matr{C}_i\) have a certain
structure:

\begin{definition} \label{def:subtoep}
  A matrix \(\matr{C}\in\field^{m\times n}\) is a \emph{sub-Toeplitz}
  matrix if \(m\le n\) and \(\matr{C}\)
  consists of $m$ rows selected out of a full $n\times n$
  Toeplitz matrix.
\end{definition}

Note that we can always write such a matrix \(\matr{C}\) as a sub-permutation
matrix \(\matr{S} \in \{0,1\}^{m \times n}\) times the full Toeplitz
matrix \(\matr{T}\in\field^{n\times n}\), i.e.,
\(\matr{C} = \matr{S}\matr{T}\)%
\footnote{We hope that the reader will forgive us for overloading the capital letter S:
a bold \(\matr{S}\) always refers to this sub-permutation matrix, while a sans-serif
\(\fieldsubset\) refers to a subset of the field \(\field\) used to select random elements.}.
The benefit for us is in the communication savings:

\begin{lemma}\label{lem:subtoepcost}
  An \(m\times n\) sub-Toeplitz matrix \(\matr{C}\) can be sent
  with \(\oh{n}\) communication.
\end{lemma}
\begin{proof}
  Writing \(\matr{C}=\matr{S}\matr{T}\) as above,
  simply send the \(2n-1\) entries of the full Toeplitz matrix
  \(\matr{T}\) and the \(m\) row indices selected by \(\matr{S}\).
\end{proof}

The number \(t\) of sub-Toeplitz matrices sent must be large enough,
according to the field size, so that the Prover can actually find them
with the required properties
(see \cref{alg:rowmemprover} below). This value \(t\) is computed
by the Verifier and Prover independently as shown in
\cref{step:rowmem:gett}, where we use the slight abuse of notation
that, when \(\field\) is infinite,
\(\log_{\#\field} \alpha = 0\) for any positive finite value $\alpha$.

\begin{protocol}
  \caption{\textsf{RowSpaceMembership}}
  \label{pro:rowmem}
  \Public{\(\matr{A}\in\F[x]^{m\times n}\),
    \(\vect{v}\in\F[x]^{1\times n}\), \(\rho \in \NN\)}
  \Certifies{\(\vect{v} \in \rowsp_{\F[x]}(\matr{A})
    \enspace\text{and}\enspace \rank(\matr{A})=\rho\)}
  \begin{protocolsteps}
  \step \label{step:rowmem:rankub}
    \subprotocol{ \(\rank(\matr{A}) \checkle \rho\) using \rankub{}}
  \vspace{0.2cm}
  \step \label{step:rowmem:gett}
    \verifier{\hspace*{-2.5 em}\(\begin{array}{@{}l@{}}
      \rho \checkle \min(m,n) \\
      t \gets 1 + \\
      \hspace{.5 em} \max\big(1, \lceil \log_{\#\field/\rho}(2\rho\deg(\matr{A}))\rceil\big)
      \\
    \end{array}\)}
  \step \label{step:rowmem:pcomp}
    \prover{\(\begin{array}{@{}l@{}}
      \text{Compute sub-Toeplitz } \matr{C}_1,\ldots,\matr{C}_t \in \F^{\rho\times m}\\
      \text{and polynomials } \dd_1,\ldots,\dd_t \in\F[x]\\
      \text{s.t. } \forall i, \rank(\matr{C}_i\matr{A})=\rho, \\
      \text{and } \forall i, \dd_i\vect{v}\in\rowsp_{\F[x]}(\matr{C}_i\matr{A}), \\
      \text{and } \gcd(\dd_1,\ldots,\dd_t) = 1
    \end{array}\)}
    \toverifier{\(\matr{C}_1,\ldots,\matr{C}_t,\dd_1,\ldots,\dd_t\)}
  \step \label{step:rowmem:degcheck}
  \verifier{
      \(\hspace*{-2.5 em}\begin{array}{@{}l@{}}
      \forall i, \deg(\dd_i) \checkle \rho\deg(\matr{A})
    \end{array}\)}
  \vspace{0.2cm} \step \label{step:rowmem:gcdcheck}
    \subprotocol{\(\gcd(\dd_1,\ldots,\dd_t)\checkeq 1\) using \progcd{}}
  \vspace{0.2cm} \step\label{step:rowmem:rs2}
  \subprotocol{
    \begin{tabular}{l}
      \textbf{for} \(i=1,\ldots,t\) \textbf{do}\\
      \hspace{1em}  \(\dd_i\vect{v} \checkin\rowsp_{\F[x]}(\matr{C}_i\matr{A})\)
        and \(\rank(\matr{C}_i\matr{A}) \checkeq \rho\)\\
      \hspace{1em} using \rowmemf{}
    \end{tabular}
  }
  \end{protocolsteps}
\end{protocol}

We now proceed to show how the Prover can actually find the values required on
\cref{step:rowmem:pcomp}. We write \(r=\rank(\matr{A})\); if the Prover is
honest, then in fact \(r = \rho\). The next lemma is inspired from
\citep{MS04}.

\begin{lemma}\label{lem:pdiv}
  Let \(\matr{A}\in\F[x]^{m\times n}\) with rank \(r\);
  \(\vect{v}\in\rowsp_{\F[x]}(\matr{A})\);
  \(\matr{S}\in\{0,1\}^{r \times m}\) be a selection of \(r\) out of
  \(m\) rows;
  \(p\in\F[x]\) be an irreducible polynomial;
  and \(\matr{T}\in\F^{m\times m}\) be a Toeplitz matrix with entries
  chosen independently and uniformly at random from a finite subset
  \fieldsubset of \F.
  Then either \(\matr{S}\matr{T}\matr{A}\) \emph{always} has rank
  strictly below \(r\), or for any rational solution
  \(\vect{w}\in\F(x)^{1\times r}\) to
  \(\vect{w}\matr{S}\matr{T}\matr{A}=\vect{v}\), the probability that
  \(\rank(\matr{S}\matr{T}\matr{A}) < r\) or that
  $p$ divides \(\denom(\vect{w})\) is at most
  \(r/\#\fieldsubset\).
\end{lemma}
\begin{proof}
  Let \(\matr{\hat{T}}\) be a generic \(m\times m\) Toeplitz
  matrix, defined by \(2m-1\) indeterminates
  $z_1,\ldots,z_{2m-1}$. Because \(\idMat{m}\) is an evaluation of
  \(\matr{\hat{T}}\), then clearly
  \(\rank(\matr{\hat{T}}\matr{A}) = \rank(\matr{A}) = r\),
  and furthermore \(\rank(\matr{S}\matr{\hat{T}}\matr{A})=r\) if and only if
  \(\matr{S}\) selects $r$ linearly independent rows from
  \(\matr{\hat{T}}\matr{A}\).

  So for the remainder assume that
  \(\rank(\matr{S}\matr{\hat{T}}\matr{A})\)
  is nonsingular over \((\field[x])[z_1,\ldots,z_{2m-1}]\); otherwise
  \(\rank(\matr{S}\matr{T}\matr{A}) < r\)
  for any choice of \(\matr{T}\), and we are done.

  The structure of the proof is as follows. We first show the existence of
  a unimodular-completable matrix \(\matr{U}\in\pmatRing{n}{r}\) such that
  \(\matr{S}\matr{T}\matr{A}\) and \(\matr{S}\matr{T}\matr{U}\) are closely
  related: in particular, they both have full rank \(r\) if and only if
  the latter has non-zero determinant, and \(p\) divides \(\vect{w}\) only
  when this determinant is divisible by \(p\). The proof proceeds to
  demonstrate these properties,
  as well as the fact that \(\matr{S}\matr{\hat{T}}\matr{U}\) has nonzero determinant
  generically, and therefore with high probability for a random choice of \(\matr{T}\).

  Let \(\matr{P}\in\{0,1\}^{n\times r}\) be a sub-permutation matrix
  which selects \(r\) linearly independent columns from
  \(\matr{S}\matr{\hat{T}}\matr{A}\). Then \(\rank(\matr{A}\matr{P})=r\)
  and we consider a factorization
  \(\matr{A}\matr{P} = \matr{U}\matr{B}\), where
  \begin{itemize}
    \item \(\matr{B}\in\F[x]^{r\times r}\) is a row basis of
      \(\matr{A}\matr{P}\) (and therefore \(\matr{B}\) is nonsingular);
      and
    \item \(\matr{U}\in\F[x]^{m\times r}\) can be completed to a square
      unimodular matrix, meaning there exists some matrix
      \(\matr{V}\in\F[x]^{m\times(m-r)}\) such that
      \(\det([\matr{U}\;\vert\;\matr{V}])\in\F\setminus\{0\}\).
  \end{itemize}
  Such a factorization always exists: if \(\matr{\hat{B}} \in \pmatRing{r}{n}\)
  is any row basis of \(\matr{A}\), then there is a unimodular matrix
  \([\matr{U}\;\vert\;\matr{V}] \in \pmatRing{m}{m}\) such that
  \([\matr{U}\;\vert\;\matr{V}]\transpose{[\transpose{\matr{\hat{B}}}\;\vert\;\matr{0}]}
  = \matr{U}\matr{\hat{B}} = \matr{A}\); and then
  we have \(\matr{U}\matr{B} = \matr{A}\matr{P}\)
  where \(\matr{B} \defeq \matr{\hat{B}}\matr{P}\).  It is easily verified that
  \(\matr{B}\) has full row rank and the same \(\polRing\)-row space as
  \(\matr{A}\matr{P}\); that is, \(\matr{B}\) is a row basis of
  \(\matr{A}\matr{P}\).

  From this factorization and the fact that \(\matr{B}\) is nonsingular, we know that
  \[
    \rank(\matr{A}\matr{P})
    = \rank(\matr{S}\matr{\hat{T}}\matr{A}\matr{P})
    = \rank(\matr{S}\matr{\hat{T}}\matr{U}\matr{B})
    = \rank(\matr{S}\matr{\hat{T}}\matr{U})
    = r
  \]
  over the ring \((\field[x])[z_1,\ldots,z_{2m-1}]\).

  Now because
  \([\matr{U}\;\vert\;\matr{V}]\) is unimodular, it is always nonsingular over
  the extension field \(\F[x]/\langle p\rangle\), and therefore \(\rank(\matr{U})=r\)
  over \(\F[x]/\langle p\rangle\).
  Since the entries of \(\matr{S}\matr{\hat{T}}\) do not contain \(x\)
  and from the rank condition above, this means that
  \(\matr{S}\matr{\hat{T}}\matr{U}\) is nonsingular over
  \((\field[x]/\langle p \rangle)[z_1,\ldots,z_{2m-1}]\) for any choice
  of the polynomial \(p\).

  The determinant \(\det(\matr{S}\matr{\hat{T}}\matr{U})\) is therefore
  a nonzero polynomial in \(z_1,\ldots,z_{2m-1}\) over
  \(\field[x]/\langle p \rangle\) with total degree at most \(r\).
  Then, by the DeMillo-Lipton-Schwartz-Zippel lemma, the probability that
  \(\det(\matr{S}\matr{T}\matr{U})\bmod p = 0\) is at most \(r/\#\fieldsubset\).

  Connecting this back to \(\matr{A}\), if \(p\nmid\det(\matr{S}\matr{T}\matr{U})\), then
  we have
  \[
    r = \rank(\matr{S}\matr{T}\matr{U})
    = \rank(\matr{S}\matr{T}\matr{U}\matr{B})
    = \rank(\matr{S}\matr{T}\matr{A}\matr{P})
    \le \rank(\matr{S}\matr{T}\matr{A})
    \le r,
  \]
  and hence \(\matr{S}\matr{T}\matr{A}\) has full row rank \(r\).

  Finally, we show that \(\denom(\vect{w})\) divides \(\det(\matr{S}\matr{T}\matr{U})\);
  this implies \(p \nmid \denom(\vect{w})\).
  Recall that \(\matr{B}\) is nonsingular with the same \(\polRing\)-row space
  as \(\matr{A}\matr{P}\);
  then, because \(\vect{v}\in\rowsp_{\F[x]}(\matr{A})\),
  we have \(\vect{v}\matr{P}\in\rowsp_{\F[x]}(\matr{B})\),
  so there exists \(\vect{y}\in\F[x]^{1\times r}\) such that
  \(\vect{y}\matr{B} = \vect{v}\matr{P}\).
  In addition, since \(\matr{S}\matr{T}\matr{U}\) is nonsingular, there exists an
  \emph{adjugate} matrix \(\matr{D}\in\F[x]^{r\times r}\) such that
  \(\det(\matr{S}\matr{T}\matr{U})(\matr{S}\matr{T}\matr{U})^{-1}=\matr{D}\).
  Putting these facts together, we have
  \begin{align*}
    \vect{w}\matr{S}\matr{T}\matr{A} &= \vect{v} \\
    \vect{w}\matr{S}\matr{T}\matr{A}\matr{P} &= \vect{v}\matr{P} \\
    \vect{w}\matr{S}\matr{T}\matr{U}\matr{B} &= \vect{y}\matr{B} \\
    \vect{w}\matr{S}\matr{T}\matr{U} &= \vect{y} \\
    \vect{w}\det(\matr{S}\matr{T}\matr{U}) &= \vect{y}\matr{D}.
  \end{align*}
  Because the right-hand side of the last equation has entries in
  $\F[x]$, then so does the left-land side, which means that
  \(\det(\matr{S}\matr{T}\matr{U})\) is a multiple of \(\denom(\vect{w})\).
  Hence \(\denom(\vect{w})\)
  is divisible by \(p\) only if \(\det(\matr{S}\matr{T}\matr{U})\) is
  divisible by \(p\), which we already established occurs with
  probability at most \(r/\#\fieldsubset\).
\end{proof}

Repeatedly applying the previous lemma, involving calls to the rational
linear solver of \cref{alg:linsolve}, leads to a Las Vegas randomized
algorithm for an honest Prover.

Here we require the Prover to know a
finite subset \(\fieldsubset \subseteq \F\). Because this set is never communicated
nor part of the public information, it is not necessarily the same as any subset
\(\fieldsubset\) used by the Verifier in other protocols. In order to match with
\Cref{pro:rowmem}, the Prover should choose \(\fieldsubset = \field\) if
\(\field\) is finite, and otherwise \(\#\fieldsubset \ge 2r^2 \deg(\matr{A})\).

\begin{algorithm}
  \caption{Honest Prover for \rowmem{}}
  \label{alg:rowmemprover}
  \KwIn{\(
    \matr{A}\in\F[x]^{m\times n}\) with rank \(r\), \(\vect{v}\in\rowsp_{\F[x]}(\matr{A})\),
    finite \(\fieldsubset \subseteq \F\)
    }
  \KwOut{
    \(\matr{C}_1,\ldots,\matr{C}_t,\dd_1,\ldots,\dd_t\) satisfying the
    conditions of \cref{step:rowmem:pcomp} from \cref{pro:rowmem}
  }
  \(t \gets 1+\lceil \log_{\#\fieldsubset/r}(2r\deg(\matr{A}))\rceil\)\;
  \Repeat{\(\rank(\matr{T}_1\matr{A}) = r\)\label{step:rmp:rank}}{
      \(\matr{T}_1\gets\) random \(m\times m\) Toeplitz matrix with
        entries from \fieldsubset\;
  }
  \(\matr{S}\in\{0,1\}^{r\times m} \gets \) selection of \(r\) linearly
    independent rows from \(\matr{T}_1\matr{A}\)
    \label{step:rmp:rrp}\;
  \(w_1 \gets\) solution to \(\vect{w}_1\matr{S}\matr{T}_1\matr{A}=\vect{v}\), using
        \cref{alg:linsolve} \label{step:rmp:solve1} \;
  \Repeat(\label{step:rmp:gcd1}){\(\gcd(\denom(\vect{w}_1),\ldots,\denom(\vect{w}_t))=1\)\label{step:rmp:gcd2}}{
    \(i \gets 2\)\;
    \While{\(i \le t\)\label{step:rmp:nstoep1}}{
      \(\matr{T}_i\gets\) random \(m\times m\) Toeplitz matrix with
        entries from \fieldsubset\;
      \(\vect{w}_i\gets\) solution to
        \(\vect{w}_i\matr{S}\matr{T}_i\matr{A}=\vect{v}\), using
        \cref{alg:linsolve} \;
      \lIf{\(\vect{w}_i\) is not \lowrank{}}{
        \(i \gets i+1\) \label{step:rmp:nstoep2}
      }
    }
  }
  \KwRet{\(\matr{S}\matr{T}_1,\ldots,\matr{S}\matr{T}_t\) and
  \(\denom(\vect{w}_1),\ldots,\denom(\vect{w}_t)\)}
\end{algorithm}

\begin{lemma}\label{lem:rowmemcomp}
  If \(\vect{v}\in\rowsp_{\F[x]}(\matr{A})\) and
  \(\#\fieldsubset \ge 2r\) where \(r\defeq\rank(\matr{A})\), then \cref{alg:rowmemprover}
  is a correct Las Vegas randomized algorithm with expected
  cost bound \(\tssoftoh{mnr^{\omega-2}d_{\matr{A}} + r^{\omega-1}d_{\vect{v}}}\).
\end{lemma}
\begin{proof}
  Computing the rank and the row rank profile (giving \(r\) independent
  rows) on \cref{step:rmp:rank,step:rmp:rrp} can be done
  deterministically in the stated cost bound via the
  column rank profile algorithm from \citep[Section 11]{Zhou12},
  just as was used in \cref{alg:linsolve}.

  Each matrix product \(\matr{T}_i\matr{A}\) can be explicitly
  computed in \(\softoh{mnd_{\matr{A}}}\) operations using \(\oh{nd_{\matr{A}}}\) Toeplitz-vector
  products, each done in \(\softoh{m}\) operations by relying on fast polynomial multiplication \citep[Problem 5.1]{BP94}.

  If the algorithm returns, correctness is clear from the correctness of
  \cref{alg:linsolve}.

  What remains is to prove the expected number of iterations of each
  nested loop.

  From \cref{lem:pdiv}, for each random Toeplitz matrix \(\matr{T}_i\),
  the probability that \(\matr{S}\matr{T}_i\matr{A}\) is nonsingular
  is at least \(1 - r/\#\fieldsubset \ge 1/2\). Therefore the expected
  number of iterations of the first loop is at most 2, and the expected
  number of iterations of the nested while loop until \(i\) reaches \(t\)
  is at most \(2t\).

  Write \(f_i = \denom(\vect{w}_i)\).
  To find the expected number of iterations of the outer loop on
  \crefrange{step:rmp:gcd1}{step:rmp:gcd2}, we need the probability that
  \(\gcd(\dd_1,\ldots,\dd_t)=1\) given that each
  \(\matr{S}\matr{T}_i\matr{A}\) has full row rank $r$.

  If \(\gcd(\dd_1,\ldots,\dd_t)\ne 1\), then there is some irreducible polynomial
  \(p\)
  which divides every denominator \(\dd_1,\ldots,\dd_t\). Because the
  \(\matr{T}_i\)'s are chosen independently of each other, the events ``\(p\)
  divides \(\dd_i\)'' are pairwise independent; thus, according to
  \cref{lem:pdiv}, the probability that any given irreducible polynomial
  \(p\) is such a common factor is at most \((r/\#\fieldsubset)^{t-1}\).

  The degree of $\dd_1$ is at most \(rd_{\matr{A}}\)
  since \(\matr{S}\matr{T}_i\matr{A}\) is \(r\times n\) with degree \(d_{\matr{A}}\);
  this also gives an upper bound on the number of
  distinct irreducible factors $p$ of $\dd_1$.
  Taking the union bound we see that the
  probability of \emph{any} factor being shared by all
  denominators is at most
  \[\frac{r^t d_{\matr{A}}}{(\#\fieldsubset)^{t-1}},\]
  which is at most \(\tfrac{1}{2}\) from the definition of $t$.
  Therefore the expected number of iterations of the outer loop
  is \(\oh{1}\).

  The stated cost bound follows from \cref{lem:linsolve}.
  It does not involve $t$ because we can
  see that \(t\in\oh{\log (rd_{\matr{A}})}\), which is subsumed by the soft-oh
  notation.
\end{proof}

For the sake of simplicity in presentation, and because they do not
affect the asymptotic cost bound, we have omitted a few optimizations to
the Prover's algorithm that would be useful in practice, namely:
\begin{itemize}
  \item The Prover can reduce to the full column rank case by computing
    a column rank profile of \(\matr{A}\) once at the beginning (using
    \citet[Chapter 11]{Zhou12}), and then removing corresponding non-pivot
    columns from \(\matr{A}\) and \(\vect{v}\). This does not change the
    correctness, but means that each matrix \(\matr{S}\matr{T}_i\matr{A}\) is
    square.
  \item When each \(\matr{S}\matr{T}_i\matr{A}\) is square, computing \(\vect{w}_i\)
    can be done in a simpler way than by calling \cref{alg:linsolve}, as follows: check that
    \(\matr{S}\matr{T}_i\matr{A}\) is nonsingular to confirm the rank, and then use
    a fast linear system solver to obtain \(\vect{w}_i\).
  \item The solution vectors \(\vect{w}_i\) may be re-used in the
    subprotocols \rowmemf{} confirming that each
    \(\dd_i\vect{v}\in\rowsp_{\F[x]}(\matr{S}\matr{T}_i\matr{A})\).
\end{itemize}

We conclude the section by proving \rowmem{} is complete, sound,
and efficient. As above, write
\(d_{\matr{A}}\defeq\max(1,\deg(\matr{A}))\)
and \(d_{\vect{v}}\defeq\max(1,\deg(\vect{v}))\), and let \(r=\rank(\matr{A})\).

\begin{theorem}\label{thm:rowmem}
  Assuming \(\#\fieldsubset \ge 2\min(m,n)d_{\matr{A}}\), then
  \cref{pro:rowmem} is a complete and probabilistically sound
  interactive protocol which requires \(\oh{n + md_{\matr{A}}t + d_{\vect{v}}t}\)
  communication and has Verifier cost
  \(\oh{mnd_{\matr{A}}t + nd_{\vect{v}}t}\).
  If \(\vect{v}\in\rowsp_{\F[x]}(\matr{A})\), there is a Las Vegas
  randomized algorithm for the Prover with expected cost
  \(\tssoftoh{mnr^{\omega-2}d_{\matr{A}} + r^{\omega-1}d_{\vect{v}}}\).
  Otherwise, the probability that the Verifier incorrectly accepts is at
  most
  \((3 r d_{\matr{A}} + d_{\vect{v}} + 1)/\#\fieldsubset.\)
\end{theorem}
\begin{proof}
  For the communication, note that sending each \(\matr{C}_i\) has
  communication cost \(\oh{m}\) from \cref{lem:subtoepcost}.
  Furthermore, the Verifier does
  not actually compute the products
  \(\matr{C}_i\matr{A}\), but rather uses these as a \emph{black box}
  for matrix-vector products in the two subprotocols.
  For any scalar \(\alpha\in\F\), the complexity of
  computing \(\matr{C}_i\matr{A}(\alpha)\) times any vector of scalars on
  the left or right-hand side is \(\oh{mnd_{\matr{A}}}\).

  Along with the degree conditions on each \(\dd_i\) and
  \cref{thm:rankub,lem:gcd,thm:ranklb,thm:rowmemf}, this proves the communication
  and Verifier cost claims.

  The Prover's cost comes from \cref{lem:rowmemcomp}, which
  dominates the cost for the Prover in any of the subprotocols.

  If the rank conditions being verified on
  \cref{step:rowmem:rankub,step:rowmem:rs2} are true, then all matrices
  \(\matr{C}_i\matr{A}\) have full row rank equal to the rank of \(\matr{A}\),
  that is, \(\rank(\matr{C}_i\matr{A}) = \rho = r\).
  And if the statements verified on
  \cref{step:rowmem:gcdcheck,step:rowmem:rs2} are true as well, then we have
  \(\vect{v}\in\rowsp_{\F[x]}(\matr{A})\) according to \cref{lem:denom_gcd}.
  Therefore the soundness of this protocol depends only on the
  probabilistic soundness of those subprotocols.

  For the remainder of the proof, we assume that
  \(\vect{v}\not\in\rowsp_{\F[x]}(\matr{A})\) and
  we want to know an upper
  bound on the probability that the Verifier incorrectly accepts.
  For this, we divide into cases
  depending on which subprotocol incorrectly accepted:

  \paragraph{Case 1: \(\rank(\matr{A}) > \rho\)}
  According to \cref{thm:rankub}, the probability that the Verifier
  incorrectly accepts in \rankub{} on \cref{step:rowmem:rankub} is at
  most \((rd_{\matr{A}}+1)/\#\fieldsubset\).

  \paragraph{Case 2: \(\rank(\matr{A})\le\rho\)
    and \(\gcd(\dd_1,\ldots,\dd_t)\ne 1\)}
  We know that
  each \(\deg(\dd_i)\le rd_{\matr{A}}\), where \(r\) is the true rank of \(\matr{A}\).
  By \cref{lem:gcd}, the probability that the Verifier incorrectly
  accepts in subprotocol \progcd{} is at most
  \((2\max_i(\deg(\dd_i))-1)/\#\fieldsubset\), which is at most
  \((2r d_{\matr{A}}-1)/\#\fieldsubset\).

  \paragraph{Case 3: \(\rank(\matr{A})\le\rho\)
    and \(\gcd(\dd_1,\ldots,\dd_t) = 1\)}
  Then, by \cref{lem:denom_gcd},
  there exists \(i\in\{1,\ldots,t\}\) such that
  \(\dd_i\vect{v}\not\in\rowsp_{\F[x]}(\matr{A})\), and thus
  either \(\rank(\matr{C}_i\matr{A})<\rho\) or
  \(\dd_i\vect{v}\not\in\rowsp_{\F[x]}(\matr{C}_i\matr{A})\).
  That is, the statement being verified by \rowmemf{} on the $i$th
  iteration of \cref{step:rowmem:rs2} is false.

  Because of the degree checks on \cref{step:rowmem:degcheck},
  we know that
  \(\deg(\dd_i\vect{v}) \le rd_{\matr{A}} + d_{\vect{v}}\).
  Therefore from
  \cref{thm:rowmemf}, the probability that the Verifier incorrectly
  accepts in \rowmemf{}
  is at most \((4rd_{\matr{A}} + d_{\vect{v}} + 1)/\#\fieldsubset\).

  Observe that the three cases are disjoint and cover all possibilities.
  In every case, the probability that the Verifier incorrectly accepts
  is at most that in Case 3, which proves the last claim in the
  Theorem statement.
\end{proof}

We note that it is always possible to conduct the checks on
\cref{step:rowmem:rs2} of \rowmem{}
in parallel, so that the
total number of \emph{rounds} of communication in the protocol is \(\oh{1}\).

A crucial factor in the communication and Verifier costs as seen in
\cref{thm:rowmem} is the value of \(t\), which in any case satisfies
\(t\in\oh{\log(\min(m,n))}\) due to the condition on the size of
\fieldsubset,
so this adds only a logarithmic factor to the
cost. Indeed, when the set \fieldsubset{} of field elements is
large enough, $t$ can be as small as 2.  For clarity, we state
as a corollary a condition under which this logarithmic factor can be
eliminated.

\begin{corollary}
  If
  \(\#\fieldsubset\ge 2mnd_{\matr{A}}\), then
  \cref{pro:rowmem} requires only \(\oh{n + md_{\matr{A}}+d_{\vect{v}}}\)
  communication and has Verifier cost
  \(\oh{mnd_{\matr{A}} + nd_{\vect{v}}}\).
\end{corollary}

\section{Row spaces and normal forms}%
\label{sec:rowspace_normalforms}

In this section, we use the row space membership protocol from the
previous section in order to certify the equality of the row spaces
of two matrices. Along with additional non-interactive checks by the
Verifier, this can also be applied to prove the correctness of certain
important normal forms of polynomial matrices.

\subsection{Row space subset and row basis}%
\label{sub:row_space}

We will use \rowmem{} to give a protocol for the certification of
\emph{row space subset}; by this we mean the problem of deciding whether the
row space of \(\matr{A}\) is contained in the row space of \(\matr{B}\), for
two given matrices \(\matr{A}\) and \(\matr{B}\).

Our approach is the following: take a random vector $\vect{\lambda}$ and
certify that the row space element \(\vect{\lambda} \matr{A}\) is in the row
space of \(\matr{B}\), the latter being done via row space membership
(\cref{sec:rowmem}).  We will see that taking \(\vect{\lambda}\) with
entries in the base field is enough to ensure good probability of success.

\begin{lemma}
  \label{lem:rowsp_subset_proba}
  Let \(\matr{A} \in \pmatRing{m}{n}\) and \(\matr{B} \in \pmatRing{\ell}{n}\).
  Let $R \in \{ \polRing, \fracRing \}$.
  Then the following statement holds:
  Assuming that
  \[
    \rowsp_{R}(\matr{A}) \not\subseteq \rowsp_{R}(\matr{B}),
  \]
  then the \(\field\)-vector space
  \[
    \vecspace \defeq  \left\{ \vect{\lambda} \in \matRing{1}{m} \mid
    \vect{\lambda}\matr{A} \in \rowsp_{R}(\matr{B}) \right\}
  \]
  has dimension at most \(m-1\).
  For \(\vect{\lambda}\in\matRing{1}{m}\) with entries chosen independently and uniformly at
  random from a finite subset \(\fieldsubset \subseteq \field\) then
  \(\vect{\lambda}\matr{A} \in \rowsp_{R}(\matr{B}) \) with
  probability at most \(\frac{1}{\#\fieldsubset}\).
\end{lemma}
\begin{proof}
  Suppose that the vector space $\vecspace$ has dimension at least $m$. Then
  \(\vecspace\) is the entire space \(\field^{1\times m}\), and every row of
  \matr{A} is in \(\rowsp_{R}(\matr{B})\); hence
  \(\rowsp_{R}(\matr{A}) \subseteq  \rowsp_{R}(\matr{B})\), a contradiction.
  Then the probability that the uniformly random vector \(\vect{\lambda}\) belongs to the proper
  subspace \(\vecspace \subsetneq \F^{1\times m}\) follows from \cref{lem:randinnerprod}.
\end{proof}

\begin{protocol}
  \caption{\textsf{RowSpaceSubset}}
  \label{pro:rowsubset}
  \Public{$\matr{A}\in\F[x]^{m\times n}$,  $\matr{B}\in\F[x]^{\ell \times n}$,
    \(\rho\in\NN\)}
  \Certifies{\(\rowsp_{\F[x]}(\matr{A}) \subseteq \rowsp_{\F[x]}(\matr{B})\)
    and \(\rank(\matr{B}) = \rho\)}
  \begin{protocolsteps}
  \step \label{step:rowsubset:rowmem_choosec}
  \verifier{$\vect{\lambda} \randfrom \fieldsubset^{1\times m}$}
  \step \verifier{\(\vect{v} \gets \vect{\lambda} \matr{A}\)}
    \toprover{\vect{v}}
  \step\label{step:rowsubset:rowmem}
    \subprotocol{\begin{tabular}{@{}c@{}}
      \(\vect{v} \checkin \rowsp_{\polRing}(\matr{B})\)
      and \(\rank(\matr{B})\checkeq\rho\)\\
      using \rowmem{}
    \end{tabular}}
  \end{protocolsteps}
\end{protocol}

In the following, let \(r_{\matr{A}}\) and \(r_{\matr{B}}\) denote respectively the ranks of \(\matr{A}\) and
\(\matr{B}\), and let \(d_{\matr{A}}=\max(1,\deg(\matr{A}))\) and \(d_{\matr{B}}=\max(1,\deg(\matr{B}))\).
\begin{theorem}\label{thm:rowsubset}
  \Cref{pro:rowsubset} is a
    probabilistically sound interactive protocol,
    and is complete assuming
    \(\#\fieldsubset \geq 2\ell d_{\matr{B}}\) in its subprotocols.
    It requires
    \(\oh{n + (\ell d_{\matr{B}} + d_{\matr{A}})\log(\ell)}\)
  communication and has Verifier cost
  \[
    \oh{(\ell nd_{\matr{B}} + nd_{\matr{A}})\log(\ell) + m nd_{\matr{A}}}.
  \]
  If \(\rowsp_{\F[x]}(\matr{A}) \subseteq \rowsp_{\F[x]}(\matr{B})\), there is a Las Vegas
  randomized algorithm for the Prover with expected cost
  \[
    \softoh{\ell nr_{\matr{B}}^{\omega-2}d_{\matr{B}} + r_{\matr{B}}^{\omega-1}d_{\matr{A}}
    + m n d_{\matr{A}}}.
  \]
  Otherwise, the probability that the Verifier incorrectly accepts is at
  most \[\frac{4r_{\matr{B}}d_{\matr{B}} + d_{\matr{A}} + 2}{\#\fieldsubset}.\]
\end{theorem}

\begin{proof}
  The Verifier may incorrectly accept if either
  \(\vect{\lambda}\) is such that
  \(\vect{\lambda}\matr{A} \in \rowsp_{\polRing}(\matr{B})\),
  which happens with probability at most
  \(1/\#\fieldsubset\) by \cref{lem:rowsp_subset_proba}, or the
  subprotocol \rowmem{} has incorrectly accepted.
  From \Cref{thm:rowmem}, and the union bound, we obtain the claimed
  probability bound.
\end{proof}

Repeating this check in both directions proves that two matrices have
the same row space.

\begin{protocol}
  \caption{\textsf{RowSpaceEquality}}
  \label{pro:rowspaceeq}
  \Public{$\matr{A}\in\F[x]^{m\times n}$,  $\matr{B}\in\F[x]^{\ell \times n}$,
    \(\rho\in\NN\)}
  \Certifies{\(\rowsp_{\F[x]}(\matr{A}) = \rowsp_{\F[x]}(\matr{B})\)
    and \(\rank(\matr{A}) = \rank(\matr{B}) = \rho\)}
  \begin{protocolsteps}
  \step\label{step:rowspaceeq:AB}
    \subprotocol{\begin{tabular}{@{}c@{}}
      \(\rowsp_{\polRing}(\matr{A}) \checksub \rowsp_{\polRing}(\matr{B})\)
        and \(\rank(\matr{B}) = \rho\)\\
      using \rowsubset{}
    \end{tabular}}
  \vspace{0.2cm}\step\label{step:rowspaceeq:BA}
    \subprotocol{\begin{tabular}{@{}c@{}}
      \(\rowsp_{\polRing}(\matr{B}) \checksub \rowsp_{\polRing}(\matr{A})\)
        and \(\rank(\matr{A}) = \rho\)\\
      using \rowsubset{}
    \end{tabular}}
  \end{protocolsteps}
\end{protocol}

\begin{theorem}
  \label{thm:rowspaceeq}
  Let \(r=\max(r_{\matr{A}},r_{\matr{B}})\) and \(d=\max(d_{\matr{A}},d_{\matr{B}})\).
  \Cref{pro:rowspaceeq} is a probabilistically sound interactive protocol, and
  is complete assuming \(\#\fieldsubset \geq 2\max(m d_{\matr{A}},\ell
  d_{\matr{B}})\).
  It requires
    \[\oh{(m\log(m) + \ell\log(\ell))d + n} \subset \softoh{md + \ell d + n}\]
  communication and has Verifier cost
  \[\oh{(m\log(m) + \ell\log(\ell))nd}
    \subset \softoh{mnd + \ell n d}.\]
  If \(\rowsp_{\F[x]}(\matr{A}) =\rowsp_{\F[x]}(\matr{B})\), there is a Las Vegas
  randomized algorithm for the Prover with expected cost
  \(\tssoftoh{(m+\ell) nr^{\omega-2}d  }\).
  Otherwise, the probability that the Verifier incorrectly accepts is at
  most \((4rd +  d+ 2)/\#\fieldsubset.\)
\end{theorem}

From \rowspaceeq, we deduce a protocol verifying the property that \(\matr{B}\) is
a row basis of \(\matr{A}\), that is, a matrix which has the same row space
as \(\matr{A}\) and which has full row rank.

\begin{protocol}
  \caption{\textsf{RowBasis}}
  \label{pro:rowbasis}
  \Public{$\matr{A}\in\F[x]^{m\times n}$,  $\matr{B}\in\F[x]^{\ell \times n}$}
  \Certifies{\(\matr{B}\) is a row basis of \(\matr{A}\)}
  \begin{protocolsteps}
  \step\label{step:rowbasis:AB}
    \subprotocol{\begin{tabular}{@{}c@{}}
      \(\rowsp_{\polRing}(\matr{B}) \checkeq \rowsp_{\polRing}(\matr{A})\)
        and \(\rank(\matr{B})\checkeq\ell\)\\
      using \rowspaceeq{}
    \end{tabular}}
  \end{protocolsteps}
\end{protocol}
\begin{corollary}
  Let \(r=\max(r_{\matr{A}},r_{\matr{B}})\) and
  \(d=\max(d_{\matr{A}},d_{\matr{B}})\).  Then, \Cref{pro:rowbasis} is a
  probabilistically sound interactive protocol, and is complete assuming
  \(\#\fieldsubset \geq 2\max(m d_{\matr{A}}, \ell d_{\matr{B}})\).
  It requires
  \[\oh{(m\log(m) + \ell\log(\ell))d+n} \subset \softoh{md+\ell d +n}\]
  communication and has Verifier cost
  \[\oh{(m\log(m) + \ell \log(\ell))nd} \subset \softoh{mnd + \ell nd} . \]
  If \(\matr{B}\) is a row basis of \(\rowsp_{\F[x]}(\matr{A})\), there is a Las Vegas
  randomized algorithm for the Prover with expected cost
  \(\tssoftoh{mn\ell^{\omega-2}d  }\), with \(\ell=r_{\matr{A}}\) in this case.
  Otherwise, the probability that the Verifier incorrectly accepts is at
  most \((4rd + d+ 3)/\#\fieldsubset.\)
\end{corollary}

\subsection{Normal forms}
\label{sub:normal_forms}

Here, we give protocols for certifying \emph{normal forms} of polynomial
matrices, including the Hermite form \citep{Hermite1851,MacDuffee33}
and the Popov form \citep{Popov72,Kailath80}.  These forms are specific row
bases with useful properties such as being triangular for the former or having
minimal degrees for the latter, and being unique in the sense that a given
matrix in \(\pmatRing{m}{n}\) has exactly one row basis in Hermite
(resp.~Popov) form.

Roughly speaking, the Hermite form is a row echelon form that stays
within the underlying ring.

\begin{definition}
  \label{def:hermite}
  A matrix \(\hermite = [b_{i,j}] \in \pmatRing{r}{n}\) with \(r\le n\) is in
  \emph{Hermite form}  if there are \emph{pivot indices}
  \(1 \le k_1 < \cdots < k_r \le n\)
  such that:
  \begin{enumerate}[(i)]
    \item \label{item:hermite:monic}
      (Pivots are monic, hence nonzero)

      \(b_{i,k_i}\) is monic for all
      \(1\le i\le r\),

    \item \label{item:hermite:triangular}
      (Entries right of pivots are zero)

      \(b_{i,j} = 0\) for all \(i\le i\le r\) and \(k_i < j \le n\),

    \item \label{item:hermite:normalized}
      (Entries below pivots have smaller degree)

      \(\deg(b_{i',k_i})<\deg(b_{i,k_i})\) for all \(1\le i < i' \le r\).
  \end{enumerate}
\end{definition}

Each entry at row $i$ and column \(k_i\) is called a \emph{pivot}.
Observe that these conditions guarantee $\hermite$ has full row rank,
hence the use of the notation $r$ for the row dimension.
For a matrix \(\matr{A} \in \pmatRing{m}{n}\), its Hermite form \(\hermite \in
\pmatRing{r}{n}\) is the unique row basis of \(\matr{A}\) which is in Hermite
form.

Protocol \ishermite{} certifies that a matrix
\(\hermite\in\pmatRing{\ell}{n}\) is the Hermite form of \(\matr{A}\).
It first checks that \(\hermite\) is in Hermite form,
and then it checks that \(\hermite\) and
\(\matr{A}\) have the same row space using \rowspaceeq{} from
\cref{sub:row_space}.

\begin{protocol}
  \caption{\textsf{HermiteForm}}
  \label{pro:hermite}
  \Public{\(\matr{A}\in\pmatRing{m}{n}\), \(\hermite \in\pmatRing{\ell}{n}\)}
  \Certifies{\(\hermite\) is the Hermite form of \(\matr{A}\)}
  \begin{protocolsteps}
    \step \label{step:hermite:hform}
      \verifier{Check that $\hermite$ satisfies \cref{def:hermite}}
    \vspace{0.1cm}\step\label{step:hermite:rowspaceeq}
      \subprotocol{\begin{tabular}{@{}c@{}}
        \(\rowsp_{\polRing}(\matr{A}) \checkeq \rowsp_{\polRing}(\hermite)\)
          and \(\rank(\matr{B})\checkeq\ell\)\\
        using \rowspaceeq{}
      \end{tabular}}
  \end{protocolsteps}
\end{protocol}

\begin{theorem}\label{thm:hermite}
  Let \(r=\max(r_{\matr{A}},r_{\hermite})\) and \(d=\max(d_{\matr{A}},d_{\hermite})\).
  \Cref{pro:hermite} is a probabilistically sound interactive protocol
  and is complete assuming \(\#\fieldsubset \geq 2\max(m d_{\matr{A}},\ell d_{\hermite})\)
  in its subprotocol.
  It requires
  \(
    \oh{md\log(m) + n}
  \)
  communication and has Verifier cost
  \(
    \oh{mnd\log(m)}
  \).
  If \hermite{} is the Hermite form of \matr{A}, there is a Las
  Vegas randomized algorithm for the Prover with expected cost
  \(\tssoftoh{m nr^{\omega-2}d  }\). Otherwise, the probability that the
  Verifier incorrectly accepts is at most \((4rd +  d+
  2)/\#\fieldsubset\).
\end{theorem}
\begin{proof}
  To check that \hermite{} is in Hermite form at
  \cref{step:hermite:hform}, the Verifier first computes the pivot
  indices as the index of the first nonzero on each row, then checks the
  degree conditions specified in \cref{def:hermite}. (If any row is
  zero, \hermite{} is not in Hermite form.) This is a deterministic
  check with complexity only \(\oh{\ell n}\).

  As discussed previously, the fact that \hermite{} is in Hermite form
  immediately implies that it has full row rank
  $\ell$, and hence checking the row space equality is sufficient to
  confirm that \hermite{} is a row basis for \matr{A}.

  The subprotocol \rowspaceeq{} dominates the complexity and is also the only
  possibility for the Verifier to incorrectly accept when the statement
  is false; hence the stated costs follow directly from
  \cref{thm:rowspaceeq}.
\end{proof}

While the Hermite form has an echelon shape, it is also common in polynomial
matrix computations to resort to the Popov form, for which the pivot of a row
is no longer the rightmost nonzero entry but rather the rightmost entry whose
degree is maximal among the entries of that row. This form loses the
echelon shape, but has the advantage of having smaller-degree
entries than the Hermite form.

Here we consider the more
general \emph{shifted} forms \citep{BarBul92,BeLaVi06}, which encompass Hermite
forms and Popov forms via the use of the following degree measure. For a given
tuple \(\shift = (s_1,\ldots,s_n) \in \ZZ^n\), the \(\shift\)-degree of the row
vector \(\vect{v} = [v_1 \; \cdots \; v_n] \in \pmatRing{1}{n}\) is
\[
  \deg_{\shift}(\vect{v}) = \max(\deg(v_1)+s_1,\ldots,\deg(v_n)+s_n).
\]
We use the notation \(\popov_{i,*}\) to denote the $i$th row of the
matrix \popov.

\begin{definition}
  \label{def:spopov}
  Let \(\shift = (s_1,\ldots,s_n) \in \ZZ^n\). A matrix \(\popov = [b_{i,j}]
  \in \pmatRing{r}{n}\) with \(r\le n\) is in \emph{\(\shift\)-Popov form} if
  there are indices \(1 \le k_1 < \cdots < k_r \le n\) such that,
  \begin{enumerate}[(i)]
    \item \label{item:popov:monic}
      (Pivots are monic and determine the row degree)

      \(b_{i,k_i}\) is monic and \(\deg(b_{i,k_i})
      + s_{k_i} = \deg_{\shift}(\popov_{i,*})\)
      for all \(1\le i\le r\),

    \item \label{item:popov:triangular}
      (Entries right of pivots do not reach the row degree)

      \(\deg(b_{i,j}) + s_{j} < \deg_{\shift}(\popov_{i,*})\)
      for all \(1\le i\le r\) and \(k_i < j \le n\),

    \item \label{item:popov:normalized}
      (Entries above and below pivots have lower degree)

      \(\deg(b_{i',k_i})<\deg(b_{i,k_i})\)
      \(1\le i' \neq i \le r\).
  \end{enumerate}
\end{definition}

The usual Popov form corresponds to the
uniform shift \(\shift = (0,\ldots,0)\).
Furthermore, one can verify that, specifying the shift as \(\shift = (n
t,\ldots,2t,t)\) for any given $t>\deg(\popov)$, then the Hermite form
is the same as the $\shift$-Popov form
\citep[Lem.\,2.6]{BeLaVi06}.

For a matrix \(\matr{A} \in \pmatRing{m}{n}\), there exists a unique row basis
\(\popov \in \pmatRing{r}{n}\) of \(\matr{A}\) which is in \(\shift\)-Popov
form \citep[Thm.\,2.7]{BeLaVi06}; \(\popov\) is called the
\emph{\(\shift\)-Popov form of \(\matr{A}\)}. Generalizing
\cref{pro:hermite} to this more general normal form yields
\cref{pro:spopov} (although the former could be derived as a particular case
of the latter for a specific shift \(\shift\), we preferred to write both
explicitly for the sake of clarity).

\begin{protocol}
  \caption{\textsf{ShiftedPopovForm}}
  \label{pro:spopov}
  \Public{\(\matr{A}\in\pmatRing{m}{n}\), \(\shift=(s_1,\ldots,s_n)\in\ZZ^n\), \(\popov\in\pmatRing{\ell}{n}\)}
  \Certifies{\(\popov\) is the \(\shift\)-Popov form of \(\matr{A}\)}
  \begin{protocolsteps}
    \step \verifier{Check that $(\shift,\popov)$ satisfies \cref{def:spopov}}
    \vspace{0.1cm}\step\label{step:spopov:rowbasis}
      \subprotocol{\begin{tabular}{@{}c@{}}
        \(\rowsp_{\polRing}(\matr{A}) \checkeq \rowsp_{\polRing}(\hermite)\)
          and \(\rank(\matr{B})\checkeq\ell\)\\
        using \rowspaceeq{}
      \end{tabular}}
  \end{protocolsteps}
\end{protocol}

The next result is identical to \cref{thm:hermite}, in both
statement and proof. The only difference in the protocol is determining the indices of
each pivot column in order to confirm the conditions of $\shift$-Popov
form; this can be accomplished in linear time by first computing the
$\shift$-degree of the row and then finding the rightmost column which
determines this shifted row degree.

\begin{theorem}
  Let
  \(r=\max(r_{\matr{A}},r_{\popov})\) and \(d=\max(d_{\matr{A}},d_{\popov})\).
  \Cref{pro:spopov} is a probabilistically sound interactive protocol
  and is complete assuming
  \(\#\fieldsubset \geq 2\max(m d_{\matr{A}},\ell d_{\popov})\) in its
  subprotocol. It requires
  \(
    \oh{md\log(m) + n}
  \)
  communication and has Verifier cost
  \(
    \oh{mnd\log(m)}
  \).
  If \(\popov\) is the \(\shift\)-Popov form of \(\matr{A}\), there is a Las
  Vegas randomized algorithm for the Prover with expected cost
  \(\tssoftoh{m nr^{\omega-2}d  }\). Otherwise, the probability that the
  Verifier incorrectly accepts is at most \((4rd +  d+
  2)/\#\fieldsubset\).
\end{theorem}

\section{Saturation and kernel bases}
\label{sec:saturation}

In this section, we use the protocols described in previous sections to
design protocols verifying computations related to saturations and kernels of
polynomial matrices.

\subsection{Saturation and saturated matrices}

The saturation of a matrix over a principal ideal domain is a useful tool in
computations; we refer to \citep[Section II.\S2.4]{BourbakiCommAlg} for a general
definition of saturation. It was exploited for example in \citep{ZhoLab13} where a
matrix is factorized as the product of a column basis times some saturation
basis, and in \citep{NRS18} in order to find the location of pivots in the
context of the computation of normal forms. The
saturation can be computed from the Hermite form, as described in
\citep[Section 8]{PS10} for integer matrices, and alternatively it can be
obtained as a left kernel basis of a right kernel basis of the matrix as
we prove below (\cref{lem:saturation_is_kernel}).

\begin{definition}
  The \emph{saturation} of a matrix $\matr{A}\in \pmatRing{m}{n}$ is the
  $\polRing$-module
  \[
    \satur (\matr{A}) = \pmatRing{1}{n} \cap \rowsp_{\fracRing}(\matr{A});
  \]
  it contains $\rowsp_{\polRing}(\matr{A})$ and has rank \(r = \rank(\matr{A})\).
  A \emph{saturation basis of \(\matr{A}\)} is a matrix in \(\pmatRing{r}{n}\)
  whose rows form a basis of the saturation of \(\matr{A}\).
  A matrix is said to be \emph{saturated} if its saturation is equal to its
  \(\polRing\)-row space.
\end{definition}

Two matrices with the same saturation may have different \(\polRing\)-row
spaces. For example, the matrices
\[
  \begin{bmatrix}
    1 & 1 \\
    \var^2 & \var^2+\var \\
    \var & \var
  \end{bmatrix}
  \quad \text{and} \quad
  \begin{bmatrix}
    1 & 1+\var^2 \\
    0 & \var^2
  \end{bmatrix}
\]
have the same saturation \(\pmatRing{1}{2}\), but the \(\polRing\)-row space of
the former matrix contains $[0 ~ \var]$ which is not in the \(\polRing\)-row
space of the latter matrix. We also remark that all nonsingular matrices in
\(\pmatRing{n}{n}\) have saturation equal to \(\pmatRing{1}{n}\).

The saturation is defined in terms of the \(\fracRing\)-row space of the
matrix: two matrices have the same saturation if and only if they have the same
\(\fracRing\)-row space. In particular, \(\matr{A}\) is saturated if and only
if any row basis of \(\matr{A}\) is saturated. This yields the following
characterization for matrices having full column rank.

\begin{lemma}
  \label{lem:saturated_fullcolrank}
  Let \(\matr{A} \in \pmatRing{m}{n}\) have full column rank. Then \(\matr{A}\)
  is saturated if and only if \(\rowsp_{\polRing}(\matr{A}) =
  \pmatRing{1}{n}\).
\end{lemma}
\begin{proof}
  Since \(\matr{A}\) has full column rank, its row bases are nonsingular \(n
  \times n\) matrices, or equivalently, \(\rowsp_{\fracRing}(\matr{A}) =
  \fracRing^{1 \times n}\).  Hence the saturation of \(\matr{A}\) is \(\pmatRing{1}{n}\),
  and the equivalence follows by definition of being saturated.
\end{proof}

Thus, in this case, verifying that \(\matr{A}\) is saturated boils down to
verifying that \(\pmatRing{1}{n}\) is a subset of
\(\rowsp_{\polRing}(\matr{A})\), which can be done using \rowsubset.

To obtain a similar result in the case of matrices with full row rank, we will
rely on the following characterization of the saturation using kernel bases.

\begin{lemma}
  \label{lem:saturation_is_kernel}
  Let \(\matr{A}\in\pmatRing{m}{n}\) have rank \(r\), and let
  \(\matr{K}\in\pmatRing{n}{(n-r)}\) be a basis for the right kernel of \(\matr{A}\).
  Then, \(\satur(\matr{A})\) is the left kernel of \(\matr{K}\). In particular,
  the saturation bases of \(\matr{A}\) are precisely the left kernel bases of
  \(\matr{K}\).
\end{lemma}
\begin{proof}
  Each row of \(\matr{A}\) is in the left kernel of \(\matr{K}\), hence so is
  any polynomial vector \(\vect{v} \in \pmatRing{1}{n}\) which is an
  \(\fracRing\)-linear combination of rows of \(\matr{A}\), that is, any
  \(\vect{v} \in \satur(\matr{A})\).

  For the other direction, it is enough to prove that each row of a given left kernel basis
  \(\matr{B}\in\pmatRing{r}{n}\) of \(\matr{K}\) is in \(\satur(\matr{A})\).
  Let \(\matr{\hat{A}}\in\pmatRing{r}{n}\) be a set of \(r\) linearly independent
  rows of \(\matr{A}\); since these rows are in the left kernel of
  \(\matr{K}\), we have \(\matr{\hat{A}}=\matr{U}\matr{B}\) for some nonsingular
  \(\matr{U}\in\pmatRing{r}{r}\). Thus each row of
  \(\matr{B}=\matr{U}^{-1}\matr{\hat{A}}\) is an \(\fracRing\)-linear
  combination of rows of \(\matr{A}\).
\end{proof}

Combining this with \citep[Lemma~3.3]{ZhoLab13}, it follows that for any saturation
basis \(\matr{B}\in\pmatRing{r}{n}\) of \(\matr{A}\) and any factorization
\(\matr{A} = \matr{C} \matr{B}\) with \(\matr{C}\in\pmatRing{m}{r}\), then
\(\matr{C}\) is a column basis of \(\matr{A}\). If \(\matr{A}\) has full row
rank we obtain that \(\matr{C}\) is nonsingular, and that \(\matr{A}\) is
saturated if and only if \(\matr{C}\) is unimodular, or equivalently
\(\colsp_{\polRing}(\matr{A}) = \pmatRing{m}{1}\). For the sake of
completeness, we now present a concise proof of this characterization
(\cref{lem:saturated_fullrowrank}); we will need the following standard result
which essentially says that any kernel basis is saturated (see for example
\citep[Lemma~2.2]{gn18} for a proof).

\begin{fact}
  \label{lem:kerbas_unimodular_colbas} 
  Let \(\matr{K} \in \pmatRing{n}{\ell}\). For any left kernel basis \(\matr{B}
  \in \pmatRing{r}{n}\) of \(\matr{K}\), we have \(\colsp_{\polRing}(\matr{B})
  = \pmatRing{r}{1}\).
\end{fact}

\begin{lemma}
  \label{lem:saturated_fullrowrank}
  Let \(\matr{A} \in \pmatRing{m}{n}\) have full row rank. Then, \(\matr{A}\)
  is saturated if and only if \(\colsp_{\polRing}(\matr{A}) =
  \pmatRing{m}{1}\).
\end{lemma}
\begin{proof}
  If \(\matr{A}\) is saturated, it is a basis of its own saturation since it has
  full row rank. Then writing \(\matr{K}\) for a right kernel basis of
  \matr{A}, by \cref{lem:saturation_is_kernel}, \matr{A} is a left kernel
  basis of \matr{K}. Then \cref{lem:kerbas_unimodular_colbas} gives
  \(\colsp_{\polRing}(\matr{A}) = \pmatRing{m}{1}\).

  Conversely, assume
  \(\colsp_{\polRing}(\matr{A}) = \pmatRing{m}{1}\). Since the row space of
  \(\matr{A}\) is a submodule of its saturation, we have \(\matr{A} = \matr{U}
  \matr{B}\) where \(\matr{B} \in \pmatRing{m}{n}\) is a saturation basis of
  \(\matr{A}\) and \(\matr{U} \in \pmatRing{m}{m}\) is nonsingular. By
  assumption, we have \(\matr{A} \matr{V} = \idMat{m}\) for some \(\matr{V} \in
  \pmatRing{n}{m}\), hence \(\matr{U}(\matr{B} \matr{V}) = \idMat{m}\).
  Because these are all polynomial matrices, this means that
  \(\matr{U}\) is unimodular, and \(\matr{A} = \matr{U} \matr{B}\) implies that
  \(\matr{A}\) is saturated.
\end{proof}

We are now ready to state \cref{pro:saturated} for the certification that a
matrix is saturated, assuming it has either full row rank or full column rank.
The latter restriction is satisfied in all the applications we have in mind,
including the two we present below (\cref{sub:kernel_completion}): unimodular completability and kernel basis
certification. We note that, if one accepts a communication cost similar to the size
of the public matrix \(\matr{A}\), then removing this assumption is easily done by
making use of a row basis of \(\matr{A}\).

\begin{protocol}
  \caption{\textsf{Saturated}}
  \label{pro:saturated}
  \Public{\(\matr{A}\in\pmatRing{m}{n}\)}
  \Certifies{\(\matr{A}\) is saturated and \(\matr{A}\) has full rank}
  \begin{protocolsteps}
  \step\label{step:is_saturated:bli}
     \subprotocol{
       \begin{minipage}{0.8\textwidth}
         \textbf{if} \(m\le n\): \hfill \texttt{// full row rank} \linebreak
         \phantom{If} \(\pmatRing{m}{1} \checksub \colsp_{\polRing}(\matr{A})\)
                      and \(\rank(\matr{A}) \checkeq m\) \hfill{} \linebreak
         \phantom{If} using \rowsubset{} with public matrices \(\idMat{m}\) and \(\transpose{\matr{A}}\) \hfill{} \linebreak
         \textbf{else}: \hfill \texttt{// full column rank} \linebreak
         \phantom{If} \(\pmatRing{1}{n} \checksub \rowsp_{\polRing}(\matr{A})\)
                      and \(\rank(\matr{A}) \checkeq n\) \hfill{} \linebreak
         \phantom{If} using \rowsubset{} with public matrices \(\idMat{n}\) and \(\matr{A}\) \hfill{}
       \end{minipage}
     }
  \end{protocolsteps}
\end{protocol}

\begin{theorem}
  \label{thm:saturated}
  Let \(d = \max(1,\deg(\matr{A}))\), \(\mu=\max(m,n)\),
  and \(\nu=\min(m,n)\). \Cref{pro:saturated} is a probabilistically sound
  interactive protocol and is complete assuming \(\#\fieldsubset \geq 2 \mu d\) in
  its subprotocol. It requires \(\oh{\mu d \log\mu}\) communication
  and has Verifier cost \(\oh{m n d \log\mu}\). Assuming that
  \(\matr{A}\) has full rank and is saturated, there is a Las
  Vegas randomized algorithm for the Prover with expected cost \(\tssoftoh{\mu
    \nu^{\omega-1} d}\), and otherwise the probability that the Verifier
    incorrectly accepts is at most \((4 \nu d + 2) / \#\fieldsubset\).
\end{theorem}
\begin{proof}
  This directly follows from
  \cref{lem:saturated_fullrowrank,lem:saturated_fullcolrank} and
  \cref{thm:rowsubset}. Remark that in both cases \(m\le n\) and \(m > n\),
  the protocol \rowsubset{} is applied with public matrices \(\idMat{\nu}\)
  and a \(\mu\times\nu\) matrix of rank at most \(\nu\) and degree at
  most \(d\).
\end{proof}

Concerning the certification of a saturation basis of \(\matr{A}\), our
protocol will rely on the following characterization.

\begin{lemma}
  \label{lem:saturation_basis}
  Let \(\matr{A}\in\pmatRing{m}{n}\). Then, a matrix \(\matr{B}\in\pmatRing{\ell}{n}\)
  is a saturation basis of \(\matr{A}\) if and only if the following conditions
  are satisfied:
  \begin{enumerate}[(i)]
    \item \(\rowsp_{\fracRing}(\matr{A}) \subseteq \rowsp_{\fracRing}(\matr{B})\),
    \item \(\rank(\matr{B}) = \ell\) and \(\rank(\matr{A}) \ge \ell\),
    \item \(\matr{B}\) is saturated.
  \end{enumerate}
\end{lemma}
\begin{proof}
  If \(\matr{B}\) is a saturation basis of \(\matr{A}\), then by definition
  \(\matr{B}\) is saturated; \(\matr{B}\) has full row rank with \(\ell =
  \rank(\matr{B}) = \rank(\matr{A})\); and \(\rowsp_{\fracRing}(\matr{B}) = \rowsp_{\fracRing}(\matr A)\).

  Conversely, assume that the three items hold. The first two items together
  imply \(\ell = \rank(\matr{B}) = \rank(\matr{A})\), and hence $\rowsp_{\fracRing}(\matr{A}) = \rowsp_{\fracRing}(\matr{B})$.
  This means \(\satur(\matr{A}) = \satur(\matr{B})\), and the latter saturation is equal to \(\rowsp_{\polRing}(\matr{B})\) since
  \(\matr{B}\) is saturated by the third item.
  Hence \(\matr{B}\) has full row rank and \(\polRing\)-row space equal to the saturation of
  \(\matr{A}\).
\end{proof}

\begin{protocol}
  \caption{\textsf{SaturationBasis}}
  \label{pro:satbasis}
  \Public{\(\matr{A}\in\pmatRing{m}{n}\), \(\matr{B}\in\pmatRing{\ell}{n}\)}
  \Certifies{\(\matr{B}\) is a saturation basis of \(\matr{A}\)}
  \begin{protocolsteps}
  \step\label{step:satbasis:ok_dims}
    \verifier{Check \(\ell \checkle \min(m,n)\)}
  \step\label{step:satbasis:right_rank}
    \subprotocol{Check \(\rank(\matr{A}) \checkge \ell\) using \ranklb}
  \step\label{step:satbasis:vecsubset_1}
    \verifier{%
      $\begin{array}{@{}l@{}}
          \alpha \randfrom \fieldsubset \\
          \vect{\lambda} \randfrom \fieldsubset^{1 \times n} \\
          \vect{v} \leftarrow \vect{\lambda}\matr{A}(\alpha) \in \F^{1 \times n}
      \end{array}$
    }
    \toprover{$\alpha, \vect{v}$}
  \step\label{step:satbasis:vecsubset_2}
    \prover{\begin{tabular}{@{}l@{}}
      Find \(\vect{\gamma}\in\F^{1 \times n}\) such that\\
      \(\vect{\gamma}\matr{B}(\alpha) = \vect{v}\).
    \end{tabular}}
    \toverifier{$\vect{\gamma}$}
  \step\label{step:satbasis:vecsubset_3}
    \verifier{
      $\vect{\gamma}\matr{B}(\alpha) \checkeq \vect{v}$
     }
  \step\label{step:satbasis:is_saturated}
    \subprotocol{Check that \(\matr{B}\) is saturated using \saturated}
  \end{protocolsteps}
\end{protocol}

\begin{theorem}
  \label{thm:satbasis}
  \Cref{pro:satbasis} is a probabilistically sound interactive protocol and is
  complete assuming \(\#\fieldsubset \geq \max(\ell d_{\matr{A}}+1,
  2 n d_{\matr{B}})\) in its subprotocols. It requires \(\oh{n
    d_{\matr{B}} \log(n)}\) communication and
    has Verifier cost
  \(
    \oh{m n d_{\matr{A}} + \ell n d_{\matr{B}} \log(n)}.
  \)
  If \(\matr{B}\) is a saturation basis of \(\matr{A}\), then there is a Las
  Vegas randomized algorithm for the Prover with expected cost
  \[
    \softoh{m n \ell^{\omega-2} + m n d_{\matr{A}}
      + n \ell^{\omega-1} d_{\matr{B}}};
  \]
  otherwise the probability that the Verifier incorrectly accepts is at most $(4 \ell d_{\matr{B}}+2)/\#\fieldsubset$.
\end{theorem}
\begin{proof}
  The check on \cref{step:satbasis:ok_dims} has no arithmetic cost, but
  ensures that \(\ell\) is less than or equal to $m$ and $n$.
  \Cref{step:satbasis:vecsubset_1,step:satbasis:vecsubset_2,step:satbasis:vecsubset_3} certify that $\rowsp_{\fracRing}(\matr{A}) \subseteq \rowsp_{\fracRing}(\matr{B})$.
  Note that this only communicates $O(n)$ field elements, and that the Verifier's cost at these steps amounts to the evaluation of $\matr{A}$ and $\matr{B}$ at $\alpha$ as well as two scalar vector-matrix products, hence a total of \(\tsoh{m n d_{\matr{A}} + \ell n d_{\matr{B}}}\) operations.
  Then the Verifier and Prover's costs follow from \cref{thm:ranklb,thm:saturated}.

  Note that \(\rank(\matr{A})\ge\ell\) and
  \(\rowsp_{\polRing}(\matr{A})\subseteq\rowsp_{\polRing}(\matr{B})\)
  imply that \(\rank(\matr{B})=\ell\), so that the precondition for
  the \saturated{} protocol on
  \cref{step:satbasis:is_saturated} is valid unless one of the previous
  checks failed.

  For the probability bound, we consider each of the checks in \cref{step:satbasis:right_rank,step:satbasis:vecsubset_3,step:satbasis:is_saturated}.
  By \cref{thm:ranklb}, the Verifier incorrectly accepts an $\matr A$ with rank $< \ell$ with probability at most $1/{\#\fieldsubset}$.
  By \cref{thm:saturated}, the Verifier incorrectly accepts an unsaturated $\matr B$ with probability at most $(4\ell d_{\matr B}+2)/\#\fieldsubset$.

  For \cref{step:satbasis:vecsubset_3}, assume now that $\rank(\matr{A}) \geq \ell$ and $\matr B$ is saturated, and in particular $\rank(\matr B) = \ell$, but that $\rowsp_{\fracRing}(\matr{A}) \not\subset \rowsp_{\fracRing}(\matr B)$.
  By \cref{lem:rowsp_subset_proba}, then for a random $\vect\lambda \in \fieldsubset^{1 \times m}$ we have $\vect\lambda \matr{A} \in \rowsp_{\fracRing}(\matr B)$ with probability at most $\frac 1 {\#\fieldsubset}$.
  Consider now that $\vect\lambda\matr{A} \notin \rowsp_{\fracRing}(\matr B)$.
  Let $\matr P$ be an $n \times n$ permutation matrix such that $\matr B\matr P = [ \matr B_0 \mid \matr B_1 ]$ with $\matr B_0 \in \polRing^{\ell \times \ell}$ and full rank.
  Let $\vect u \in \fracRing^{1 \times \ell}$ be the unique vector such that $\vect\lambda \matr{A} P = [ \vect u \matr B_0 \mid \vect{\hat a} ]$, for some $\vect{\hat a} \in \fracRing^{1 \times (n-\ell)}$.
  Then there is some index $i$ of $\vect{\hat a}$ that differs from index $i$ of $\vect u \matr B_1$.
  By Cramer's rule the entries of $\vect u$ can be written with common denominator $\det(B)$ and numerators of degree at most $(\ell-1)d_{\matr B}$.
  Hence the entries of $\vect u \matr B_1$ have denominators and numerators of degree at most $\ell d_{\matr B}$ each.
  Hence $\vect\lambda \matr A(\alpha)$ and $\vect u(\alpha) \matr B(\alpha)$ agree at index $i$ for at most $2\ell d_{\matr B}$ choices of $\alpha \in \field$, and we conclude that the Verifier incorrectly accepts in this case with probability at most $2\ell d_{\matr B}/\#\fieldsubset$.

  Summing up, the worst case is the check at \cref{step:satbasis:is_saturated}, where the Verifier incorrectly accepts with probability at most $(4\ell d_{\matr B} + 2)/\#\fieldsubset$.
\end{proof}

\subsection{Kernel bases and unimodular completability}
\label{sub:kernel_completion}

Here, we derive two protocols which follow from the ones concerning the
saturation. The second protocol is for the certification of kernel bases, while
the first protocol is about matrices that can be completed into unimodular
matrices.

The fast computation of such completions was studied by \citet{ZhoLab14}. We
say that \(\matr{A} \in \pmatRing{m}{n}\) is \emph{unimodular completable} if
\(m<n\) and there exists a matrix \(\matr{B}\in \pmatRing{(n-m)}{n}\) such that
\(\begin{bmatrix} \matr{A} \\ \matr{B} \end{bmatrix}\) is unimodular. Note that
if \(\matr{A}\) does not have full row rank, then it is not unimodular
completable. Otherwise, \citet[Lemma 2.10]{ZhoLab14} showed that \(\matr{A}\)
is unimodular completable if and only if \(\matr{A}\) has unimodular column
bases; by \cref{lem:saturated_fullrowrank}, this holds if and only if
\(\matr{A}\) is saturated. This readily leads us to \cref{pro:completable}.

\begin{protocol}
  \caption{\textsf{UnimodularCompletable}}
  \label{pro:completable}
  \Public{\(\matr{A}\in\pmatRing{m}{n}\)}
  \Certifies{\(\matr{A}\) is unimodular completable}
  \begin{protocolsteps}
  \step\label{step:is_completable:right_dims}
    \verifier{\(m \checklt n\)}
  \vspace{0.2cm}\step\label{step:is_completable:saturated}
    \subprotocol{\begin{tabular}{@{}c@{}}
      Check that \(\matr{A}\) is saturated
        and \(\rank(\matr{A})\checkeq m\)\\
      using \saturated{}
    \end{tabular}}
  \end{protocolsteps}
\end{protocol}

\begin{theorem}
  \label{thm:completable}
  \Cref{pro:completable} is a probabilistically sound interactive protocol and
  is complete assuming \(\#\fieldsubset \geq 2 m d\) in its subprotocols.
  It
  requires \(\oh{n d \log(n)}\) communication and has Verifier cost \(\oh{m n d
  \log(n)}\). If \(\matr{A}\) is unimodular completable, then there is a Las
  Vegas randomized algorithm for the Prover with expected cost \(\tssoftoh{n
  m^{\omega-1} d}\); otherwise the probability that the Verifier
  incorrectly accepts is at most \((4 m d + 2) / \#\fieldsubset\).
\end{theorem}
\begin{proof}
  The costs follow from \cref{thm:ranklb,thm:saturated}, noting that
  the protocol aborts early if \(m\ge n\), and therefore \(m\) is an
  upper bound on the rank in the \saturated{} subprotocol.
  The probability of the
  Verifier incorrectly accepting here is the same as in \saturated{}
  from \cref{thm:saturated}.
\end{proof}

Finally, \cref{pro:kernel} for the certification of kernel bases will follow
from the characterization in the next lemma.

\begin{lemma}
  \label{lem:kernel_basis}
  Let \(\matr{A}\in\pmatRing{m}{n}\) and let \(\matr{B}\in\pmatRing{\ell}{m}\). Then,
  \(\matr{B}\) is a left kernel basis of \(\matr{A}\) if and only if
  \begin{enumerate}[(i)]
    \item \(\rank(\matr{B}) = \ell\) and \(\rank(\matr{A}) \ge m-\ell\),
    \item \(\matr{B}\matr{A} = \zeromat\),
    \item \(\matr{B}\) is saturated.
  \end{enumerate}
\end{lemma}
\begin{proof}
  If \(\matr{B}\) is a left kernel basis of \(\matr{A}\), then we have
  \(\rank(\matr{B}) = \ell = m-\rank(\matr{A})\) as well as \(\matr{B}\matr{A} =
  \zeromat\); the third item follows from
  \cref{lem:kerbas_unimodular_colbas,lem:saturated_fullrowrank}.

  Now assume that the three items hold. Consider some left kernel basis
  \(\matr{K}\) of \(\matr{A}\). Then, \(\rank(\matr{K}) = m-\rank(\matr{A}) \le
  \ell\) by the first item, while the second item implies that the row space of
  \(\matr{B}\) is contained in the row space of \(\matr{K}\), hence \(\ell =
  \rank(\matr{B}) \le \rank(\matr{K})\); therefore \(\rank(\matr{K})=\ell\). As a
  result, \(\matr{B} = \matr{U} \matr{K}\) for some nonsingular \(\matr{U} \in
  \pmatRing{\ell}{\ell}\). Item $(iii)$
  implies \(\colsp_{\polRing}(\matr{B}) = \pmatRing{\ell}{1}\) according to \cref{lem:saturated_fullrowrank}, hence \(\idMat{\ell} =
  \matr{B} \matr{V} = \matr{U}\matr{K}\matr{V}\) for some \(\matr{V} \in
  \pmatRing{m}{\ell}\). Then, \(\matr{U}\) must be unimodular, and thus \(\matr{B}
  = \matr{U} \matr{K}\) is a left kernel basis of \(\matr{A}\).
\end{proof}

\begin{protocol}
  \caption{\textsf{KernelBasis}}
  \label{pro:kernel}
  \Public{\(\matr{A}\in\pmatRing{m}{n}\), \(\matr{B}\in\pmatRing{\ell}{m}\)}
  \Certifies{\(\matr{B}\) is a left kernel basis of \(\matr{A}\)}
  \begin{protocolsteps}
  \step\label{step:kernel:right_dims}
    \verifier{\(\ell \checkle m\)}
  \vspace{0.2cm}\step\label{step:kernel:rankA}
    \subprotocol{\(\rank(\matr{A}) \checkge m-\ell\) using \ranklb}
  \vspace{0.2cm}\step\label{step:kernel:in_kernel}
    \subprotocol{\(\matr{B}\matr{A} \checkeq \zeromat\) using \matmul}
  \vspace{0.2cm}\step\label{step:kernel:basis_of_kernel}
    \subprotocol{\begin{tabular}{@{}c@{}}
      Check that \(\matr{B}\) is saturated
        and \(\rank(\matr{B})\checkeq \ell\)\\
      using \saturated{}
    \end{tabular}}
  \end{protocolsteps}
\end{protocol}

\begin{theorem}
  \label{thm:kernel}
  \Cref{pro:kernel} is a probabilistically sound interactive protocol and is
  complete assuming \(\#\fieldsubset \geq
  \max((m-\ell)d_{\matr{A}}+1,2md_{\matr{B}})\) in its subprotocols. It
  requires \(\oh{m d_{\matr{B}} \log(m)}\) communication and has Verifier cost
  \[
    \oh{\ell m d_{\matr{B}} \log(m) + m n d_{\matr{A}}}.
  \]
  If \(\matr{B}\) is a left kernel basis of \(\matr{A}\), then there is a Las
  Vegas randomized algorithm for the Prover with expected cost
  \[
    \softoh{m \ell^{\omega-1} d_{\matr{B}} + m n (m-\ell)^{\omega-2}
    d_{\matr{A}}};
  \]
  otherwise the probability that the Verifier incorrectly accepts is at most
  \[
    \frac{\max(d_{\matr{A}} + d_{\matr{B}} + 1, 4\ell d_{\matr{B}} + 2)}{\#\fieldsubset}.
  \]
\end{theorem}
\begin{proof}
  The costs follow from
  \cref{lem:kernel_basis,thm:ranklb,thm:polyfreivalds,thm:saturated}.
  As before, the worst case for the Verifier is that only one of the
  three checked statements is wrong, and the resulting maximum of
  probabilities comes either from \cref{step:kernel:rankA} or
  \cref{step:kernel:basis_of_kernel}.
\end{proof}

\section{Conclusion and perspectives}
\label{sec:perspectives}

We have developed interactive protocols verifying a variety of problems
concerning polynomial matrices. For rank, determinant, system solving, and
matrix multiplication (\cref{sec:vspace}), these amount to evaluating
at some random point(s) and reducing to field-based verifications. For
row bases, saturation, normal forms, and kernel basis
computations (\cref{sec:rowspace_normalforms,sec:saturation}), the
verifications essentially reduce to testing row space membership of a
single vector (\cref{sec:rowmem}) and testing that ranks are the expected ones.

Our protocols are efficient. The volume of data exchanged in communications is
roughly the size of a single row of the matrix. The time complexity for the
Verifier is linear (or nearly-linear) in the size of the object being checked,
and the time for the Prover is roughly the same as it would take to perform the
computation being verified.

Still, there is some room for improvement in these costs. It would be nice to
remove the logarithmic factors in the complexities of most later protocols for
the Verifier time and communication cost; these come from the number of
repetitions $t$ required in the \rowmem{} protocol.

Our protocols also require to work over sufficiently large fields, to ensure
soundness of the randomized verification. For smaller fields, a classic
workaround is to resort to a field extension, increasing the arithmetic and
communication cost by a logarithmic factor. An alternative is to increase the
dimension in the challenges and responses, e.g.~verifying a block of vectors
instead of a single vector of field elements. A further study on whether this
approach is applicable and competitive here is required.

Another possibility for improvement in our complexities would be to have the
same costs where $d$ is the \emph{average} matrix-vector degree, rather than
the maximum degree. Such complexity refinements have appeared for related
computational algorithms, frequently by ``partial linearization'' of the rows
or columns with highest degree \citep[Section 6]{GuSaStVa12}, and it would be
interesting to see if similar techniques could work here. This would be
especially helpful in more efficiently verifying an unbalanced shifted Popov
form, and the Hermite form in particular, of a nonsingular matrix.

The protocols presented here do not assume that the Prover has computed the result to be
verified. This is however likely to be the case in many instances of verified computing, and it
would then  be relevant to identify which
intermediate results in a Prover's computation of the solution (such as the rank profile matrix, the
weak Popov form, etc), could be reused in a certificate
for verifying this solution.
Though more constraining on the Prover's choice of an algorithm, such information would help reducing the
leading constant in the arithmetic cost of its computation.

While we have presented protocols for a variety of basic problems on
polynomial matrices, there are still more for which we do not know yet
whether any efficient verification exists. These include:
\begin{itemize}
  \setlength\itemsep{0pt}
  \item matrix division with remainder (see \citep[Section IV.\S2]{Gantmacher59} and
  \citep[Theorem 6.3-15]{Kailath80});
  \item high-order terms in expansion of the inverse
    (see the high-order lifting algorithm of \citet{Storjohann03});
  \item univariate relations, generalizing Hermite-Pad\'e approximation
    \citep{BecLab00,NeiVu17};
  \item Smith form (see \citep{Storjohann03} for the fastest known
    algorithm).
\end{itemize}
We also do not know in all cases how to prove the \emph{negation} of our
statements --- for example, that a vector is \emph{not} in the row space
of a polynomial matrix. It seems that some similar techniques to those
we have used may work, but we have not investigated the question deeply.

Perhaps the most interesting
direction for future work would be to adapt our protocols
to the case of Euclidean lattices, i.e., integer matrices and vectors.
It seems that most of our protocols in \cref{sec:vspace} should
translate when we replace evaluation at a point \(\alpha\) with
reduction modulo a sufficiently-large prime $p$, but the analysis in terms of bit
complexity rather than field operations will likely be more delicate.
Another seeming hurdle is in our central protocols in \cref{sec:rowmem}
on deciding row membership: while the general ideas of these protocols
\emph{might} translate to integer lattices, the proof techniques we have
used are particular for polynomials.

\section*{Acknowledgements}
\addcontentsline{toc}{section}{Acknowledgements}

This work was performed while the fourth author was generously hosted by
the Laboratoire Jean Kuntzmann in Grenoble.

This work was partially supported by the U.S.\ National Science
Foundation under award \#1618269,
by the
\href{http://opendreamkit.org}{OpenDreamKit}
\href{https://ec.europa.eu/programmes/horizon2020/}{Horizon 2020}
\href{https://ec.europa.eu/programmes/horizon2020/en/h2020-section/european-research-infrastructures-including-e-infrastructures}{European
  Research Infrastructures} project under award
\#\href{http://cordis.europa.eu/project/rcn/198334_en.html}{676541},
by the \href{http://www.institutfrancais.dk}{IFD-Science 2017} research program of the
Institut Français du Danemark,
and by the CNRS-INS2I Institute through its program for young researchers.

\bibliographystyle{elsarticle-harv} 

\newcommand{\Gathen}{\relax}\newcommand{\Hoeven}{\relax}

\end{document}